\documentclass[aip,graphicx, reprint]{revtex4-1}

\usepackage{amsmath}
\usepackage{comment}
\usepackage{amsthm}
\usepackage{tikz}
\usetikzlibrary{arrows.meta}
\usepackage{amssymb}
\newtheorem{theorem}{Theorem}

\newtheorem*{cor*}{Corollary}
\newtheorem{lem}[theorem]{Lemma}
\newtheorem{pro}[theorem]{Proposition}

\usepackage{bm}
\usepackage{mathtools}

\usepackage[colorlinks=true, allcolors=blue]{hyperref}

\renewcommand{\selectlanguage}[1]{} 

\begin{document}
\title{Phase and gain stability for adaptive dynamical networks}

\author{Nina Kastendiek}

\affiliation{These authors have contributed equally.}
\affiliation{Institute for Chemistry and Biology of the Marine Environment, University of Oldenburg, 26129 Oldenburg, Germany}

\author{Jakob Niehues}
\email[corresponding author: ]{jakob.niehues@pik-potsdam.de}

\affiliation{These authors have contributed equally.}
\affiliation{Potsdam Institute for Climate Impact Research (PIK),
 Member of the Leibniz Association, P.O. Box 60 12 03, D-14412 Potsdam, Germany}
 \affiliation{
Technische Universit\"at Berlin, ER 3-2, Hardenbergstrasse 36a, 10623 Berlin, Germany}

\author{Robin Delabays}

\affiliation{School of Engineering, University of Applied Sciences of Western Switzerland HES-SO, Sion, Switzerland}

\author{Thilo Gross}

\affiliation{Institute for Chemistry and Biology of the Marine Environment, University of Oldenburg, 26129 Oldenburg, Germany}
\affiliation{Helmholtz Institute for Functional Marine Biodiversity (HIFMB), 26129 Oldenburg, Germany}
\affiliation{Alfred Wegener Institute (AWI), Helmholtz Center for Polar and Marine Research, 27570 Bremerhaven, Germany}

\author{Frank Hellmann}

\affiliation{Potsdam Institute for Climate Impact Research (PIK),
 Member of the Leibniz Association, P.O. Box 60 12 03, D-14412 Potsdam, Germany}

\date{\today}

\begin{abstract}
In adaptive dynamical networks, the dynamics of the nodes and the edges influence each other. We show that we can treat such systems as a closed feedback loop between edge and node dynamics. Using recent advances on the stability of feedback systems from control theory, we derive local, sufficient conditions for steady states of such systems to be linearly stable. These conditions are local in the sense that they are written entirely in terms of the (linearized) behavior of the edges and nodes.

We apply these conditions to the Kuramoto model with inertia written in adaptive form, and the adaptive Kuramoto model. For the former we recover a classic result, for the latter we show that our sufficient conditions match necessary conditions where the latter are available, thus completely settling the question of linear stability in this setting.
The method we introduce can be readily applied to a vast class of systems. It enables straightforward evaluation of stability in highly heterogeneous systems.
\end{abstract}


\maketitle

\textbf{
In complex adaptive networks, both the nodes and the connections between them evolve and adapt over time. A fundamental question is, how these evolving interactions influence the system's stability. This paper introduces a new method to analyze such adaptive dynamical systems by treating them as a closed feedback loop between node and edge dynamics. By drawing on recent developments in control theory, we derive a theorem on sufficient conditions for the linear stability of the interconnected system based on the properties of its components. The theorem provides concise local (node-wise and edge-wise) stability conditions. In contrast to standard results, we do not have to assume that all nodes and edges are the same, and our conditions are largely independent of details of the connectivity. We apply these conditions to paradigmatic models in network dynamics and show that they provide new insights and confirm existing results. This approach offers a powerful and accessible tool for evaluating the stability of a wide range of complex systems, including those with highly heterogeneous components or equilibria.
}

\section{Introduction}
Many systems can be conceptualized as networks in which the network topology is evolving while there are also simultaneously dynamics in the network nodes. If these types of dynamics interact, a feedback loop between local and topological dynamics is formed, and the system can be called an \emph{adaptive network} \cite{anreview}.  

Adaptive networks can exhibit rich dynamics and complex self-organization as the network topology can act as a memory that is shaped by past dynamics and thus effectively increases the dimensionality of the phase space. Dynamics of adaptive networks play a role in a wide range of phenomena including social distancing in epidemics \cite{anepi,marceau,scarpino}, opinion formation processes in humans \cite{kozma,Hegselmann,sanmiguel} and animals\cite{couzin}, strategic interactions\cite{skyrms,do}, neural self-organization\cite{bornholdt,meisel,kuehn}, and ecology \cite{ecorev}, among many others \cite{bernerrev}.    

Due to the dynamical interplay between state and topology, mathematical analysis of the dynamics of adaptive networks is difficult. Hence, much of the earlier literature in the field focuses on discrete-state adaptive networks that can be modeled by systems of ordinary differential equations using moment expansions \cite{demirel}.  
More recently, the master stability function approach\cite{levin,pecora_master_1998} has been generalized to broad classes of adaptive networks\cite{berner_desynchronization_2021, berner_synchronization_2021}. However, the use of master stability functions hinges on the symmetry between network nodes. Hence, this approach only allows for stability and bifurcation analysis of homogeneous states of adaptive networks. 

In this work, we introduce a new method to analyze the stability of heterogeneous adaptive dynamical networks by leveraging recent advances in linear algebra and control theory. The key ingredient is to write the system in the form of a feedback loop between node and edge variables, represented by transfer matrices with a block diagonal structure. The central new tool that enables this analysis is the concept of the phase of a matrix \cite{wang_phases_2020}. This gives information complementary to the information in the singular values. It was already observed in \citet{wang_phases_2023} that these phases reveal interesting information on the Laplacian of a graph.

The central result of \citet{chen_phase_2024} is that the phases of transfer matrices can be used to give sufficient stability conditions for interconnected systems. \citet{zhao_when_2022} observed that phase and gain information can be combined to cover a much broader class of systems. In this paper, we combine the observation that phase analysis is well-behaved for Laplacian-like systems, with the results of \citet{chen_phase_2024,zhao_when_2022} to provide novel sufficient stability conditions for adaptive dynamical systems.
Furthermore, we leverage the natural block structure of adaptive networks to obtain local instead of global conditions.

We apply these conditions to the paradigmatic adaptive Kuramoto model, and find that the sufficient stability condition we provide, matches the necessary stability condition of \citet{do_topological_2016} where the latter is applicable. Together, these results completely characterize the stable steady state configurations of the adaptive Kuramoto model. This demonstrates that, despite their generality, the conditions are not very conservative in this important special case.
Furthermore, we recover standard results for the classical Kuramoto model, and the Kuramoto model with inertia.
A companion paper \cite{niehues_small-signal_2024} develops the necessary formulations and results to apply these methods to complex oscillator models of power grids \cite{kogler_normal_2022, buttner_complex_2024}.

Our paper proceeds as follows:
In \ref{sec:high level sketch}, we provide an intuitive introduction to the concepts of controlling feedback systems.
In \ref{subsec:transfer functions}, we recall important fundamental notions in control theory.
Readers familiar with control theory can skip these two sections, and start with section \ref{subsec:response MIMO} which gives the central definitions we need for our paper, and \ref{subsec:feedback stability MIMO LTI} which gives the recent results on the feedback stability of systems that we build upon.
In \ref{sec:stability in adaptive networks}, we give our main results by adapting these concepts to the analysis of adaptive dynamical networks.
In \ref{sec:applications} we apply them to example systems.
We recover state-of-the-art results and extend them to heterogeneous systems.

Throughout the paper we write capital bold letters for matrices, e.g., $\bm M$, and lower case bold letters for vectors, e.g., $\bm x$.

\section{Phase stability for dynamical systems: High-level sketch}
\label{sec:high level sketch}
Here we give an intuitive introduction to the concepts that we later treat rigorously.

Consider the set of equations:
\begin{align}
    \dot x &= - b \cdot y \label{eq:physicist-toy-model 1}\\
    y &= d \cdot x \label{eq:physicist-toy-model 2}
\end{align}
These are two systems that are connected by treating the ``output" of the one as the ``input", or driving force, of the other.

Such interconnected linear systems are studied in depth in control theory. One system is often considered a ``plant", that is the system we wish to control using the inputs, and the second is a feedback controller. The interconnected system is called the closed loop system.

One of the foundational results of the theory of the stability of such feedback systems is the small gain theorem \cite{zames_input-output_1966}. If we drive each of the systems with an oscillation of a fixed frequency, the response will be a phase shifted oscillation with a different amplitude. In the case of \eqref{eq:physicist-toy-model 1}:
\begin{align}
    \dot x &= - b \cdot y \; :\\
    y = \exp(i \omega t) &\Rightarrow x = \frac{- b}{i \omega} \exp(i \omega t) + c\\
    &\Rightarrow x = \sigma \exp(i \omega t + \phi) + c
\end{align}
where $\sigma = \left|\frac{b}{\omega}\right| > 0$ is called the gain and $\phi = \arg\left(\frac{- b}{i \omega}\right)$ is called the phase of the response.

The gain at each frequency is (the absolute value of) the factor by which the oscillation is amplified/damped. The small gain theorem\cite{zames_input-output_1966} states, that if the product of the gains of the two interconnected systems is smaller than $1$ at all frequencies, the interconnected system is stable.

The gain of the second equation \eqref{eq:physicist-toy-model 2} is just $|d|$, however, the gain of the first equation is $\left|\frac{b}{\omega}\right|$ and becomes arbitrarily large. Consequently, small gain theorems can not prove stability of such connected systems.

Allowing $b$ and $d$ to be complex, the condition for stability is of course $\Re(bd) > 0$. This result can be recovered by using the phases of the system under study. Small phase theorems give conditions of the type:
\begin{equation}
    - \pi  < \arg(d) + \arg\left(\frac{b}{i \omega}\right) < \pi \quad \forall \omega \neq 0 
\end{equation}
which reduces to the exact result. A strong advantage of this condition is that it is invariant under scaling the coefficients $b$ and $d$ by positive values.

The paradigmatic Kuramoto model \cite{kuramoto_self-entrainment_1975} can be written in the same structure as \eqref{eq:physicist-toy-model 1}, \eqref{eq:physicist-toy-model 2}. To linear order, the driving force from the neighboring nodes is provided by a weighted Laplacian multiplying the phases $\bm x$.
Again, we know linear stability conditions for the simplest case: If the weighted Laplacian $\bm L$ is positive semi-definite, then $\dot{\bm x} = - \bm L \bm x$ is semi-stable. In particular, this is the case if the weights are positive.

The small gain theorem was stated for higher dimensional systems from the start, with the singular values of matrices taking the place of $\sigma$ above \cite{zames_input-output_1966}. However, in Kuramoto-type models, we typically have that the uncoupled system can move freely along a limit cycle. There is no local force pushing it towards a particular phase, and thus the gain of the nodal dynamics will often be infinite. Small gain theorems do not apply for the same reason as in \eqref{eq:physicist-toy-model 1}, \eqref{eq:physicist-toy-model 2}.

In contrast to gain theorems, phase conditions for the stability of higher dimensional systems were only derived recently \cite{chen_phase_2024}. They build on a new definition of the phases of a matrix \cite{wang_phases_2020}. These definitions allow treating systems like the linearized Kuramoto oscillator that are not amenable to small gain analysis. Intriguingly, the phases on which these small phase theorems rely are also invariant under rescaling in the following sense: They are matrix functions that satisfy $\phi(\bm A) = \phi(\bm M \bm A \bm M^\dagger)$ for full rank, square $\bm M$.

As we will see, these types of relationships allow us to reduce the phase conditions on the weighted Laplacian $\bm L = \bm B \bm D \bm B^\dagger$ to phase conditions on the diagonal $\bm D$. For the simple Kuramoto model, the condition that the weights have to be positive is easily recovered. Thus, the phase analysis subsumes the fact that positive weights imply semi-definite Laplacians.

In the context of phase analysis, the network weights do not need to be constant, though. They can react to the inputs (as mapped to the edges by $\bm B^\dagger$) and the phase response of their dynamics will determine the stability characteristics of the full interconnected system.

\section{Control Theory}
\label{sec:control theory}
To make this paper more self-contained for a physics audience, this section reproduces key concepts and results from control theory. We begin by recalling key properties of transfer functions in the one-dimensional case. We then discuss the definition of the magnitude and phase response of higher dimensional linear systems (in this context typically called multiple input, multiple output, linear time-invariant: MIMO LTI), and the underlying concepts of matrix phase and magnitude. Finally, we present the phase and gain conditions for stability in interconnected feedback systems that we will use below.

\subsection{Transfer functions and stability}
\label{subsec:transfer functions}

Transfer functions are a tool for understanding how linear systems respond to inputs or disturbances, thus providing insights into their behavior and stability. We give an introduction to stability analysis with the help of transfer functions. For more detailed explanations, refer, for example, to \citet{levine_control_2018, bechhoefer_control_2021}. 

The transfer function of an LTI system relates the system's input to its output in the Laplace domain. The \textit{Laplace transform} converts a function $f(t)$ of a real variable $t$, typically in the time domain, to a function $F(s)$ of a complex variable $s$ in the complex-valued frequency domain, also called $s$-domain 
\begin{align}
\mathcal{L}\{f(t)\} = F(s) = \int_0^{\infty} e^{-st} f(t) \, dt.
\end{align}
The complex frequency variable $s$ is defined as $s = \rho + i\omega$, where the real part $\rho$ is related to growth or decay and the imaginary part $i\omega$ corresponds to oscillatory components.
Following conventions, we will often use $f$ for the function $f(t)$ and its Laplace transform $F(s)$. The context will make it clear which one is at consideration.

The \textit{transfer function} is defined as the ratio of the Laplace transform of the output $Y(s)$ to the Laplace transform of the input $U(s)$, i.e.,  $H(s) = \frac{Y(s)}{U(s)}$.
It is convention to assume without loss of generality that all initial conditions are zero.

As an example, consider the linear system
\begin{align}
\label{eq:example ode}
\dot x(t) = -a x(t) + u(t)
\end{align}
with initial conditions equal to zero: $x(0) = 0$. Taking the Laplace transform, we have
\begin{align}
s X(s) &= -a X(s) + U(s)\;,\\
X(s) &= \frac{1}{s+a} U(s)\;.
\end{align}
With $x(t)$ as the output, the transfer function $G(s)$ of the system is therefore
\begin{align}
\frac{X(s)}{U(s)} &= G(s) = \frac{1}{s+a}.
\end{align}
Setting the denominator to zero, $s+a=0$ shows that the system has a pole at $s=-a$.
The inverse Laplace transform of $G(s)$ is
\begin{align*}
\mathcal{L}^{-1}\left\{\frac{1}{s+a}\right\}=e^{-at}.
\end{align*}
The inverse Laplace transform turns the multiplication with $G(s)$ into convolution in the time domain:
\begin{align}
x(t) &= \int_0^t e^{-a (t-t')} u(t') dt',
\end{align}
which is the well-known solution of \eqref{eq:example ode} for $x(0)=0$.

A \textit{real-rational} transfer function can be expressed as $G(s) = \frac{N(s)}{D(s)}$, where $N(s)$ and $D(s)$ are polynomials in $s$ with real coefficients. In a \textit{proper} transfer function, the degree of the numerator polynomial $N(s)$ does not exceed the degree of the denominator polynomial $D(s)$, ensuring that the function does not grow unbounded as $|s|$ approaches infinity. Real-rational proper transfer functions are exactly those that correspond to dynamical systems as above.

Values of $s$ that satisfy $N(s) = 0$ define the zeros of the system, whereas values of $s$ that satisfy $D(s) = 0$ correspond to the poles of the system.

For a continuous-time LTI system represented by a transfer operator to be \textit{stable}, all poles must have negative real parts, i.e., they must lie in the left half of the complex plane. In the example above, this corresponds to $a > 0$. This ensures that, for any bounded input, the system's output remains bounded, known as Bounded Input, Bounded Output (BIBO) stability. A system that is \textit{semi-stable} (or marginally stable) can have poles on the imaginary axis, provided no poles exist in the right half of the plane. Any pole with a positive real part indicates instability. The set of \textit{real-rational stable proper} transfer functions is denoted ${\cal RH}_\infty$.

\begin{figure}
\centering
\begin{tikzpicture}
    \tikzstyle{block} = [draw, rectangle, minimum height=3em, minimum width=4.7em]
    \node [block] (G) at (5.5,0){$G(s)$};
    \draw[-{Latex[length=3mm]}] (3.5,0) -- (G);
    \draw[-{Latex[length=3mm]}] (G) -- (7.5,0);
\end{tikzpicture}
\caption{Open-loop system. Arrows indicate inputs and outputs.}
\label{fig:openloop}
\end{figure}
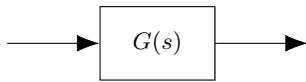
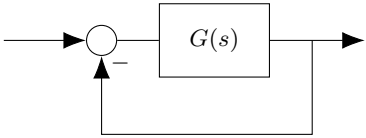
\begin{figure}
\centering
\begin{tikzpicture}
    \tikzstyle{block} = [draw, rectangle, minimum height=3em, minimum width=4.5em]
    \tikzstyle{sum} = [draw, circle, inner sep=1pt, minimum size=0.4cm]
    \node [block] (G) at (5.5,0){$G(s)$};
    \node [sum] (sum) at (4,0){};
    \draw[-] (sum) -- (G);
    \draw[-{Latex[length=3mm]}] (2.7,0) -- (sum);
    \draw[-{Latex[length=3mm]}] (G) -- (7.5,0);
    \draw[{Latex[length=3mm]}-] (sum) |- (6,-1.25) -| (6.8,0);
    \node at (4.25, -0.3) {$-$};
\end{tikzpicture}
\caption{Closed-loop system. The minus next to the circle indicates the subtraction of external input and output to implement negative feedback.}
\label{fig:closedloop}
\end{figure}

Adding feedback to a system changes the pole locations, which can have stabilizing or destabilizing effects. Even if the open-loop system (see Fig. \ref{fig:openloop}) is stable, feedback can introduce new dynamics that may destabilize the system. A closed-loop system with negative feedback (see Fig. \ref{fig:closedloop}) subtracts the feedback signal from the input signal. The transfer function of such a system is given by $T(s) = \frac{G(s)}{1+G(s)}$. 

If we add a feedback controller with transfer function $H(s)$, as in Fig. \ref{fig:closedloop2}, the closed-loop transfer function becomes $T(s) = \frac{G(s)}{1+G(s)H(s)}$.
Finding the poles of a closed-loop transfer function is more challenging than for an open-loop system, particularly in systems with multiple inputs and outputs.
Instead, the open-loop frequency response $G(s) H(s)$ can be used to determine whether the closed-loop system will be stable. The frequency response evaluates the steady-state behavior of a system by analyzing how it reacts to sinusoidal inputs across the frequency spectrum $s=i\omega$. When a sine wave is passed through a linear system, the long-term response will be also be a sine wave with the same frequency but possibly different amplitude and phase.  Initially, transient effects may occur; however, if the open-loop transfer function is stable (i.e., its poles lie in the stable region), these transients will decay over time, leaving the steady-state response.

The gain of a system, $\sigma(G(i\omega)) \coloneqq |G(i\omega)|$, describes how the system changes the amplitude of signals at a specific frequency, whereas the phase $\phi(G(i\omega)) \coloneqq \arg(G(i\omega))$ captures the phase shift introduced by the system to those signals. For real-valued linear time-invariant systems, analyzing the frequency response over $\omega \in [0, \infty]$ is sufficient, as the behavior at negative frequencies is symmetric (magnitude even and phase odd).

The Nyquist criterion can be used to give a condition for the stability of such a closed-loop system: Assume that $G(s)$ and $H(s)$ have no poles in the right half-plane, then if $1+G(s)H(s)$ has no zeros in the right half-plane, the closed loop transfer function is stable. By integrating along a contour encircling the right half-plane, and using the Cauchy argument principle, we find that $1 + G(s) H(s)$ can have no zeros in the right half-plane, if $G(s)H(s)$ does not encircle $-1$. A sufficient condition for no encirclement is that either the gain of $G(s)H(s)$ is smaller than $1$ or that its argument does not cross the negative imaginary axis: $-\pi < \arg(G(s) H(s)) < \pi$.

\begin{figure}
    \centering
    \begin{tikzpicture}
        \tikzstyle{block} = [draw, rectangle, minimum height=3em, minimum width=4.5em]
        \tikzstyle{sum} = [draw, circle, inner sep=1pt, minimum size=0.4cm]
        \node [block] (G) at (2,1) {$G(s)$};
        \node [block] (H) at (2,-0.75) {$H(s)$};
        \node [sum] (sum) at (0,1){};
        \draw[-{Latex[length=3mm]}] (-1.5,1) -- (sum);
        \draw[-{Latex[length=3mm]}] (H) -| (sum);
        \draw[-{Latex[length=3mm]}] (sum) -- (G);
        \draw[-{Latex[length=3mm]}] (4,1) |- (H);
        \draw[-{Latex[length=3mm]}] (G) -- (5,1);
        \node at (0.3, 0.75) {$-$};
    \end{tikzpicture}  
    \caption{Closed loop-system with feedback controller.}
    \label{fig:closedloop2}
\end{figure}
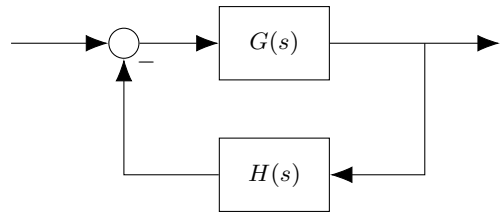

\subsection{Frequency response of MIMO LTI systems}
\label{subsec:response MIMO}

In multiple input multiple output systems, the transfer function is a matrix that describes the relationship between each input and each output. The central challenge in generalizing the arguments above to this case, is that we now must give conditions for the eigenvalues of $\bm G(s) \bm H(s)$ to not encircle $-1$. While bounding the magnitude of the eigenvalues of $\bm G(s) \bm H(s)$ in terms of the matrix norms of $\bm G(s)$ and $\bm H(s)$ is straightforward, bounding the phases of the eigenvalues in terms of properties of $\bm G(s)$ and  $\bm H(s)$ is not. This section defines the required vocabulary to state the theorems which we build upon.

\paragraph{Gain of a matrix:} We begin by introducing the appropriate complex matrix magnitudes and phases, following \citet{chen_phase_2024, zhao_when_2022}.
A matrix ${\bm M} \in \mathbb {C}^{n \times n}$ has $n$ magnitudes, defined as the $n$ singular values. The \textit{gain} is the maximum singular value and thus equal to the spectral norm.

\paragraph{Angular numerical range:} The \textit{numerical range} of a matrix $\bm M \in \mathbb {C}^{n \times n}$ is given by
\begin{align}
    W({\bm M}) = \left\{{\bm z}^\dagger {\bm M}{\bm z} ~|~ {\bm z}\in\mathbb{C}^N\, ,~ {\bm z}^\dagger {\bm z} = 1\right\}.
\end{align}
The numerical range is a compact and convex subset of $\mathbb {C}$, and contains the spectrum of ${\bm M}$. The \textit{angular numerical range} of ${\bm M}$ is defined as
\begin{align}
    W'({\bm M}) = \left\{{\bm z}^\dagger {\bm M}{\bm z} ~|~ {\bm z}\in\mathbb{C}^N\, ,~ {\bm z}^\dagger {\bm z} > 0\right\},
\end{align}
which is the conic hull of  $W({\bm M})$ and is always a convex cone (or the entire complex plane).

\paragraph{Sectorial matrices and their sector:} If $0$ is not in $W(\bm M)$, the matrix is \textit{sectorial}. If $0$ is on the boundary of $W(\bm M)$ it is \textit{semi-sectorial}. For a semi-sectorial matrix that is not identical to zero, $\arg(W'(\bm M))$ defines a sector of the circle of length $2 \Delta(\bm M) \leq \pi$. The phase for the midpoint of this sector is only defined modulo $2 \pi$, we typically choose the phase of the midpoint $\gamma(\bm M)$, such that $- \pi \leq \gamma(\bm M) < \pi$. The maximum and minimum phase of $\bm M$ are then defined as $\overline{\phi} = \gamma(\bm M) + \Delta(\bm M)$ and $\underline \phi(\bm M) = \gamma(\bm M) - \Delta(\bm M)$. Note that this means that $\overline\phi$ ($\underline \phi$) can be larger (smaller) than $\pi$ ($-\pi$).

\paragraph{Frequency-wise sectorial transfer operators:} A system $\bm G \in {\cal RH}_\infty^{m \times m}$ is said to be \textit{frequency-wise sectorial} if $\bm G(i\omega)$ is sectorial for all $\omega \in [-\infty, \infty]$.

\paragraph{Semi-stable frequency-wise semi-sectorial transfer operators:} Let $\bm G$ be an $m \times m$ real rational proper semi-stable system with no poles in the open right half-plane, and $i\Omega$ the set of poles on the imaginary axis. $\bm G$ is \textit{frequency-wise semi-sectorial} if
\begin{enumerate}
    \item $\bm G(i\omega)$ is semi-sectorial for all $\omega \in [-\infty, \infty] \setminus \Omega$; and
    \item there exists an $\epsilon^* > 0$ such that for all $\epsilon^+ \leq \epsilon^*$, $\bm G(s)$ has a constant rank and is semi-sectorial along the indented imaginary axis, where the half-circle detours to the right with radius $\epsilon$ are taken around both the poles and finite zeros of $\bm G(s)$ on the frequency axis and a half-circle detour with radius $1/\epsilon^+$ is taken if infinity is a zero of $\bm G(s)$.
\end{enumerate}

\paragraph{Frequency-wise gain response of a system:} With these definitions, we can now introduce the magnitude and phase response of MIMO LTI systems, following \citet{chen_phase_2024, zhao_when_2022}. Let $\bm G$ be an $m \times m$ real rational proper transfer matrix. Then $\sigma(\bm G(i\omega))$ is the vector of singular values of $\bm G(i\omega)$, which is called the magnitude response of $\bm G$. The gain at frequency $\omega$ is the largest singular value $\overline{\sigma}(\bm G(i\omega))$.

\paragraph{Frequency-wise phase response of a system:} 
Consider a frequency-wise sectorial (semi-sectorial) system $\bm G(s)$. At zero frequency, $\bm G(0)$ (or $\bm G(\epsilon^+)$) is real. 
For a real matrix $\bm M$ the angular field of values is symmetric on the real axis. Thus $\gamma(\bm M)$ is either $0$ or $\pi$ modulo $2 \pi$. We require that the phase center of $\bm G(0)$ (or $\bm G(\epsilon^+)$) can be chosen as $0$.
In the context of feedback, this means using the freedom to assign roots of 1 to the two transfer operators without loss of generality. This enables a practical representation that avoids the critical point and the branch cut at $\pi$ and $-\pi$.
If we choose phase centers for $\bm G(i\omega)$ such that $\gamma(\bm G(i\omega))$ is continuous in the frequency $\omega \in [-\infty, \infty]$ (or on the contour avoiding the poles). Then $[\underline{\phi}(\bm G(i\omega)), \overline{\phi}(\bm G(i\omega))]$ is the sector containing the phase response of $\bm G(s)$.

Note that the gain response is an even, the phase response an odd function of the frequency $\omega$. Small gain and phase theorems are often given in terms of conditions for only $\omega \in [0, \infty]$.

\subsection{Feedback stability in MIMO LTI systems}
\label{subsec:feedback stability MIMO LTI}
We can now state the main theorems we will adapt to the adaptive network context below.

Denote the identity matrix $\bm I$. The feedback system in Fig. \ref{fig:feedback_system}, denoted $\bm G \# \bm H$, is stable if the Gang of Four matrix
\begin{align}
\bm G \# \bm H = 
\begin{bmatrix}
(\bm I + \bm H \bm G)^{-1} & (\bm I + \bm H \bm G)^{-1} \bm H \\
\bm G (\bm I + \bm H \bm G)^{-1} & \bm G (\bm I + \bm H \bm G)^{-1} \bm H
\end{bmatrix}
\end{align}
is stable, i.e., $\bm G \# \bm H \in\mathcal{RH}_\infty$.

To obtain stability conditions, we need to make sure that $\bm I + \bm H(s) \bm G(s)$ has no zero eigenvalues in the right-hand side of the complex plane. This can be achieved by limiting either the gain \cite{zhou_essentials_1998, zames_input-output_1966} or the phase \cite{chen_phase_2024} response of $\bm G$ and $\bm H$:

\begin{figure}
    \centering
    \begin{tikzpicture}
        \tikzstyle{block} = [draw, rectangle, minimum height=3em, minimum width=4.7em]
        \tikzstyle{sum} = [draw, circle, inner sep=1pt, minimum size=0.4cm]
        \node [block] (G) at (2,1) {$\bm G(s)$};
        \node [block] (H) at (2,-0.75) {$\bm H(s)$};
        \node [sum] (sum1) at (0,1) {};
        \node [sum] (sum2) at (4,-0.75) {}; 
        \draw[-{Latex[length=3mm]}] (-1.5,1) -- (sum1); 
        \draw[-{Latex[length=3mm]}] (sum1) -- (G);
        \draw[-{Latex[length=3mm]}] (G) -| (sum2);
        \draw[-{Latex[length=3mm]}] (5.5,-0.75) -- (sum2); 
        \draw[-{Latex[length=3mm]}] (sum2) -- (H); 
        \draw[-{Latex[length=3mm]}] (H) -| (sum1); 
        \node at (0.3, 0.75) {$-$};
    \end{tikzpicture}
    \caption{Feedback system.}
    \label{fig:feedback_system}
\end{figure}
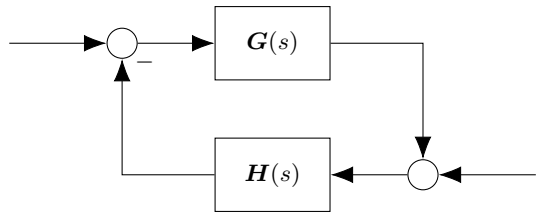

\begin{theorem}[Small Gain Theorem, \citet{zames_input-output_1966}, \citet{zhou_essentials_1998}]
Let $\bm G$ and $\bm H$ $\in\mathcal{RH}_\infty^{n \times n}$. Then the feedback system $\bm G\# \bm H$ is stable if
\begin{align}
\overline\sigma[\bm G(i\omega)]\overline\sigma[\bm H(i\omega)]<1,
\end{align}
for all $\omega \in [-\infty,\infty]$.
\end{theorem}

\begin{theorem}[Generalized Small Phase Theorem, \citet{chen_phase_2024}]
    Let $\bm G$ be semi-stable frequency-wise semi-sectorial with $i\Omega$ being the set of poles on the imaginary axis and $\bm H \in\mathcal{RH}_\infty^{n \times n}$ be frequency-wise sectorial. Then $\bm G\# \bm H$ is stable if
    \begin{align}
    \overline{\phi}(\bm G(i\omega)) + \overline\phi(\bm H(i\omega)) &< \pi,\\
    \underline{\phi}(\bm G(i\omega)) + \underline\phi(\bm H(i\omega)) &> -\pi.
    \end{align}
    for all $\omega \in [0,\infty]\backslash\Omega$.
\end{theorem}

Both conditions can be combined. The mixed gain-phase theorem with cut-off frequency developed by \citet{zhao_when_2022} provides stability results for feedback systems that satisfy a phase condition at low frequencies and a gain condition at high frequencies.

\begin{theorem}[Mixed Gain-Phase Theorem with Cut-off Frequency, \citet{zhao_when_2022}]
\label{thm:mixed gain phase cutoff}
    Let $\omega_c \in (0, \infty)$, $\bm G$ be semi-stable frequency-wise semi-sectorial over $(-\omega_c, \omega_c)$ with $i\Omega$ being the set of poles on the imaginary axis satisfying $\max_{\omega\in\Omega}|\omega|<\omega_c$, and $\bm H \in\mathcal{RH}_\infty^{n \times n }$ be frequency-wise sectorial. Then $\bm G\# \bm H$ is stable if\\
    i) for each $\omega \in [0, \omega_c) \backslash \Omega$, it holds 
    \begin{align}
    \label{eq:mixed gain phase max phases condition}
    \overline{\phi}(\bm G(i\omega)) + \overline\phi(\bm H(i\omega)) &< \pi,\\
    \label{eq:mixed gain phase min phases condition}
    \underline{\phi}(\bm G(i\omega)) + \underline\phi(\bm H(i\omega)) &> - \pi,
    \end{align}
    ii) and for each $\omega \in [\omega_c, \infty]$, it holds 
    \begin{align}
    \label{eq:mixed gain phase, gain condition}
    \overline{\sigma}(\bm G(i\omega)) \overline\sigma(\bm H(i\omega)) < 1.
    \end{align}
\end{theorem}

\citet{woolcock_mixed_2023} give a further generalized version, and a proof that proceeds by the Nyquist criterion, but does not treat the semi-stable case, which we need in the following.

\section{Stability of Adaptive Networks}
\label{sec:stability in adaptive networks}

We now come to the main result of the paper: An adaptation of the mixed gain and phase condition to the case of adaptive dynamical networks.

We will first show in detail how such systems can be written as interconnected dynamical systems. Then we will give the main theorems.

\subsection{Adaptive networks as feedback systems}
In this section, we show that a broad class of adaptive dynamical systems can be cast into the form of a feedback system.
In the neighborhood of a steady state, the subsystems in the feedback loop are represented by two transfer operators.

Let the $N$ nodes of the network be indexed $n$ and $m$, with $1 \leq n,m \leq N$.
The $L$ edges 
are indexed by ordered pairs $l=(n,m)$, $n < m$. For any nodal quantity $\bm x^v_n$, we denote $\bm x^v$ the overall vector obtained by stacking the $\bm x^v_n$, similarly for $\bm x^e_l$. Note that, unless otherwise stated, the dimension of $\bm x^e_l$ and $\bm x^v_n$ can vary from edge to edge and node to node. We write $[\bm x]$ for the diagonal matrix with the entries of $\bm x$ on the diagonal.

Consider a general bipartite dynamical system associated to a graph with node variables $\bm{x^v}$, and edge variables $\bm{x^e}$. As noted, the nodes and edges can be heterogeneous, e.g., the states $\bm{x^v}_1$ at node $1$, can be a vector of different size and follow different dynamics than $\bm{x^v}_2$ at node $2$. However, we require that the coupling on the network occurs with respect to some observables $\bm{o^v} $ and $\bm{{o'}^e}$ that have the same dimension for all edges and nodes. The nodes and edges 
of the system are coupled with their respective inputs and outputs by a node-edge incidence matrix $\bm B$ that encodes the graph structure:
\begin{align}
    \label{eq:coupled_system_1}
    \bm{\dot{x}^v} &= \bm{f^{v}}(\bm{x^v}, \bm B \bm{{o'}^e}),\\
    \label{eq:coupled_system_2}
    \bm{o^v} &\coloneqq \bm{g^{v}}( \bm{x^v}, \bm B \bm{{o'}^e}
    ), \\
    \label{eq:coupled_system_3}
    \bm{\dot{x}^e} &= \bm{f^{e}}(\bm{x^e}, \bm B^\dagger \bm{o^v}), \\
    \label{eq:coupled_system_4}
    \bm{{o'}^e} &\coloneqq \bm{g^{e}}(\bm{x^e}, \bm B^\dagger \bm{o^v}),
\end{align}
where $\bm f$ and $\bm g$ act node-wise or edge-wise, that is
\begin{align}
    \bm{\dot{x}^v}_n = \bm{f^v}_n(\bm{x^v}_n, \bm B \bm{{o'}^e}\vert_n) \qquad \text{etc.}
\end{align}
With $\bm{o^e} \coloneqq \bm B \bm{{o'}^e}$, we have
\begin{align}
    \bm{\dot{x}^v} &= \bm{f^{v}}(\bm{x^v},\bm{o^e}),\\
    \bm{o^v} &\coloneqq \bm{g^{v}}( \bm{x^v},\bm{o^e}
    ), \\
    \bm{\dot{x}^e} &= \bm{f^{e}}(\bm{x^e},\bm B^\dagger \bm{o^v}), \\
    \bm{o^e} &\coloneqq \bm B \bm{g^{e}}(\bm{x^e},\bm B^\dagger \bm{o^v}).
\end{align}
Here, each component of $\bm{f^{v}}$ and $\bm{f^{e}}$ contains the local dynamics of the $\bm{x^v}_n$ and $\bm{x^e}_l$, and also the coupling to the output of the other half of the bipartite system.
Similarly, each component of $\bm{g^{v}}$ and $\bm{g^{e}}$ contains the local output functions of $\bm{x^v}_n$ and $\bm{x^e}_l$. 
Without loss of generality, we assume the fixed point we want to study is at the origin $\bm x^* = 0$. The linearized system has the form
\begin{align}
\bm{\dot{x}^v} &= \bm{J^{vv}} \bm{x^v} + \bm{J^{ve}} \bm{o^e},\\
\bm{o^v} &\coloneqq \bm{D^{vv}} \bm{x^v} + \bm{D^{ve}} \bm{o^e}
, \\
\bm{\dot{x}^e} &= \bm{J^{ee}} \bm{x^e} + \bm{J^{ev}} \bm B^\dagger \bm{o^v}, \\
\bm{o^e} &\coloneqq \bm B \bm{D^{ee}} \bm{x^e} + \bm B\bm{D^{ev}} \bm B^\dagger \bm{o^v},
\end{align}
where the $\bm{J^{\bullet\bullet}}$ are Jacobians of the self-coupling of $\bm{x^v}$ and $\bm{x^e}$, and the input matrices. The $\bm{D^{\bullet\bullet}}$ are feedthrough and output matrices of according dimensions.
In the most common notions of adaptive networks, we have $\bm{D^{ve}} = \bm 0$. Otherwise, this system might correspond to a differential algebraic equation, rather than a differential equation.

As the dynamics of $\bm f$ and $\bm g$ was acting locally at the nodes, the square matrices $\bm{J^{\bullet\bullet}}$ and $\bm{D^{\bullet\bullet}}$ are block diagonal:
\begin{align}
    \bm{J^{v\bullet}} &= \bigoplus_n \bm{J^{v\bullet}}_n,
    \quad \bm J^{e\bullet} = \bigoplus_l \bm{J^{e\bullet}}_l,
\end{align}
and for the $\bm{D^{\bullet\bullet}}$ accordingly. $\bigoplus$ denotes the direct sum over nodes or edges.

We now go to Laplace space and obtain the transfer operators of the system:
\begin{align}
\bm{o^v} &= \left(\bm{D^{vv}} (s - \bm{J^{vv}})^{-1} \bm{J^{ve}} 
\right)\bm{o^e}\\
&\coloneqq \bm{T^{ve}}(s) \bm{o^e}\\
\bm{o^e} &= \bm B \left(\bm{D^{ee}} (s - \bm{J^{ee}})^{-1} \bm{J^{ev}} + \bm{D^{ev}} \right) \bm B^\dagger {\bm  o^v}\\
&\coloneqq -\bm B \bm{T^{ev}}(s) \bm B^\dagger \bm{o^v}.
\end{align}

If $\bm B^\dagger$ has a kernel, there are directions in phase space that are not visible to the network interactions. Note that in this situation, the stability results are with respect to outputs that are orthogonal to this zero mode.

The inner transfer operators of these systems are block diagonal: they can be written as the direct sum of node-wise or edge-wise transfer operators,
\begin{align}
    \bm{T^{ve}} &= \bigoplus_n \bm{T^{ve}}_n,
    &\bm{T^{ev}} &= \bigoplus_l \bm{T^{ev}}_l.
\end{align}
One sees immediately that this system class has the proper form to apply proposition~\ref{thm:small_phase_theorem_block}.
However, this theorem gives us global conditions.
In the following, we derive a version that leverages the block structure to provide local conditions.

\subsection{Phases and gains for networked systems}
The proofs we give below rely on the Lemmas presented in a companion paper \cite{niehues_small-signal_2024}. Here we just collect the crucial properties of the phases and the gain that make them particularly useful for studying systems with a network structure encoded in an incidence-like matrix $\bm B$.

\paragraph{Connectivity does not enlarge the phase response:} The angular field of values of the transformed matrix ${\bm B}{\bm M}{\bm B}^\dagger$ is contained within the union of the angular field of values of the matrix $\bm M$ with $0$
\begin{align}
\label{eq:sub-field}
    W'({\bm B}{\bm M}{\bm B}^\dagger) \subseteq \left( W'({\bm M}) \cup 0\right)\; ,
\end{align}
and thus the phase sector of $\bm B \bm M \bm B^\dagger$ is contained in that of $\bm M$:
\begin{align}
    \overline{\phi}({\bm B}{\bm M}{\bm B}^\dagger) &\leq \overline{\phi}({\bm M}), & 
    \underline{\phi}({\bm B}{\bm M}{\bm B}^\dagger) &\geq \underline{\phi}({\bm M}).
\end{align}

\paragraph{Composition does not enlarge the phase response:}
For a block diagonal system ${\bm T} = \bigoplus_n{\bm T}_n$, the numerical range is the convex hull of the numerical ranges of the blocks
\begin{align}\label{eq:conv-hull}
    W({\bm T}) &= {\rm Conv}\left(W({\bm T}_1),...,W({\bm T}_N)\right).
\end{align}
Therefore, 
\begin{align}\label{eq:phase bounds block matrix}
    \overline{\phi}({\bm T}) &= \max_n\overline{\phi}({\bm T}_n)\, , & 
    \underline{\phi}({\bm T}) &= \min_n\underline{\phi}({\bm T}_n).
\end{align}

\paragraph{The gain increase from connectivity is bounded:} The spectral norm is submultiplicative \cite{horn_matrix_2012}, so we can give an upper bound for the spectral norm of the transformed matrix ${\bm B}{\bm M}{\bm B}^\dagger$:
\begin{align}
\label{eq:gain_BB}
\overline{\sigma}({\bm B}{\bm M}{\bm B}^\dagger) &\leq \overline{\sigma}({\bm B^{\dagger}})\overline{\sigma}({\bm B})\overline{\sigma}({\bm M}) = \overline{\sigma}({\bm B})^2 \overline{\sigma}({\bm M}).
\end{align}

Defining the unweighted Laplacian matrix $\bm L$ of the interconnection matrix $\bm B$ by ${\bm L} = {\bm B} {\bm B}^{\dagger}$ and noting that $\overline\sigma(\bm B)^2 = \overline\sigma(\bm L)$, we get
\begin{align}
\label{eq:gain_laplacian_bb}
\overline{\sigma}({\bm B}{\bm M}{\bm B}^\dagger) &\leq \overline{\sigma}({\bm L})\overline{\sigma}({\bm M}).
\end{align}
The largest singular value of the Laplacian $\overline\sigma (\bm L)$ is equal to the largest eigenvalue and bounded by twice the maximal degree.

\paragraph{Composition does not increase the gain:} The spectral norm of a block diagonal system ${\bm T} = \bigoplus_i{\bm T}_i$, is the maximum of the spectral norms of its diagonal blocks:
\begin{align}\label{eq:gain_blocks}
    \overline{\sigma}({\bm T}) &= \max_i  \overline{\sigma}(\bm T_i).
\end{align}
\subsection{Small phase and small gain for stability in adaptive networks}
We are now ready to give the main results of this paper, a theorem that guarantees the stability of adaptive dynamical systems using gain and phase information. This generalizes the Small Phase Theorem with Block Structure given in the companion paper \cite{niehues_small-signal_2024}.

\begin{theorem}[Generalized Small Phase Theorem with Block Structure\cite{niehues_small-signal_2024}]
\label{thm:small_phase_theorem_block}
    Consider the system $\bm G\# \bm H$ with the block structure $\bm H= \bigoplus_{n=1}^N{\bm T}_n(s)$ and $\bm G = {\bm B}\bigoplus_{l=1}^{L}{\bm T}_l(s){\bm B}^\dagger$ for some $\bm B$ of appropriate dimensions.
    For each $n$, let ${\bm T}_n(s)\in{\cal RH}_\infty$ be frequency-wise sectorial. 
    For each $l$, let ${\bm T}_l(s)$ be semi-stable frequency-wise semi-sectorial, with $i\Omega$ being the union of the set of poles on the imaginary axis. Require that $\bm H(s)$ has constant rank on the indented imaginary axis. 
    Then, the interconnected system $\bm G \# \bm H$ is stable if 
    \begin{align}
    \label{eq:small phase theorem, block structure version, sectoriality condition G}
        \max_n\overline{\phi}\left({\bm T}_n(i\omega)\right) - \min_n\underline{\phi}\left({\bm T}_n(i\omega)\right) &< \pi\, , \\ 
        \max_l\overline{\phi}\left({\bm T}_l(i\omega)\right) - \min_l\underline{\phi}\left({\bm T}_l(i\omega)\right) &\leq \pi\, ,
        \label{eq:small phase theorem, block structure version, sectoriality condition H}
    \end{align}
    for all $\omega\notin\Omega$, and 
    \begin{align}
    \label{eq:small phase theorem, block structure version, stability condition sup}
        \sup_{n,l,\omega\notin\Omega}\left[\overline{\phi}\left({\bm T}_n(i\omega)\right) + \overline{\phi}\left({\bm T}_l(i\omega)\right)\right] &< \pi\, , \\ 
        \inf_{n,l,\omega\notin\Omega}\left[\underline{\phi}\left({\bm T}_n(i\omega)\right) + \underline{\phi}\left({\bm T}_l(i\omega)\right)\right] &> -\pi\, .
        \label{eq:small phase theorem, block structure version, stability condition inf}
    \end{align}
\end{theorem}

We now proceed to our main result. In analogy to the Mixed Gain-Phase Theorem with Cut-off Frequency as a combination of the Small Gain and Small Phase Theorems, we can prove a mixed theorem for systems with block structure.

\begin{pro}[Mixed Gain-Phase Theorem with Cut-off Frequency and Block Structure]
\label{thm:mixed gain phase cutoff, block}
    Consider the system $\bm G\# \bm H$ with the block structure $\bm H = \bigoplus_{n=1}^N{\bm T}_n(s)$ and $\bm G = {\bm B}\bigoplus_{l=1}^{L}{\bm T}_l(s){\bm B}^\dagger$ for some $\bm B$ of appropriate dimensions.
    Let $\omega_c \in (0, \infty)$. 
    For each $n$, let ${\bm T}_n(s)\in{\cal RH}_\infty$ be frequency-wise sectorial. 
    For each $l$, let ${\bm T}_l(s)$ be semi-stable frequency-wise semi-sectorial over $(-\omega_c, \omega_c)$, with $i\Omega$ being the union of the set of poles on the imaginary axis satisfying $\max_{\omega\in\Omega}|\omega|<\omega_c$. Further require that ${\bm T}_l(s)$ have full rank along the indented imaginary axis.
    Then, the interconnected system $\bm G \# \bm H$ is stable if\\
    \begin{align}
    \label{eq:mixed gain phase, block, sectoriality H}
        \max_n\overline{\phi}\left({\bm T}_n(i\omega)\right) - \min_n\underline{\phi}\left({\bm T}_n(i\omega)\right) &< \pi\, ,
    \end{align}
    for all $\omega\notin\Omega$, and
    \begin{align}
    \label{eq:mixed gain phase, block, semi-sectoriality cutoff G}
        \max_l\overline{\phi}\left({\bm T}_l(i\omega)\right) - \min_l\underline{\phi}\left({\bm T}_l(i\omega)\right) &< \pi\, ,
    \end{align}
    for $\omega \in (-\omega_c, \omega_c)\backslash\Omega$, and\\
    i) for each $\omega \in [0, \omega_c) \backslash \Omega$, it holds 
    \begin{align}
    \label{eq:mixed gain phase, block, phase condition sup}
        \sup_{n,l}\left[\overline{\phi}\left({\bm T}_n(i\omega)\right) + \overline{\phi}\left({\bm T}_l(i\omega)\right)\right] &< \pi\, , \\ 
        \label{eq:mixed gain phase, block, phase condition inf}
        \inf_{n,l}\left[\underline{\phi}\left({\bm T}_n(i\omega)\right) + \underline{\phi}\left({\bm T}_l(i\omega)\right)\right] &> -\pi\, .
    \end{align}
    ii) and for each $\omega \in [\omega_c, \infty]$, it holds 
    \begin{align}
    \label{eq:mixed gain phase block, gain condition}
    \sup_{n,l}\left[\overline{\sigma}(\bm B)^2\overline{\sigma}({\bm T}_n(i\omega)) \overline\sigma({\bm T}_l(i\omega))\right] < 1.
    \end{align}
\end{pro}
\textit{Remark:}
$\bm H$ is stable, and its sectoriality is ensured by \eqref{eq:mixed gain phase, block, sectoriality H}. $\bm G$ is semi-stable, and its semi-sectoriality over $(-\omega_c, \omega_c)$ is ensured by \eqref{eq:mixed gain phase, block, semi-sectoriality cutoff G}. Note that the full rank condition, together with semi-sectoriality, implies that $\bm T_l(s)$ are actually sectorial. Due to the kernel of $\bm B^\dagger$, $\bm G$ will still only be semi-sectorial.
Equations~\eqref{eq:mixed gain phase, block, phase condition sup}-\eqref{eq:mixed gain phase, block, phase condition inf} imply the phase conditions of Theorem~\ref{thm:mixed gain phase cutoff}, and \eqref{eq:mixed gain phase block, gain condition} implies the gain condition \eqref{eq:mixed gain phase, gain condition}.

\begin{proof}
    We prove this proposition by showing that all assumptions and conditions allow us to apply Theorem \ref{thm:mixed gain phase cutoff}.
    In the appendix, we prove (Lemma \ref{lem:const rank}) that the condition that $\bm T_l(s)$ has full rank, together with \eqref{eq:mixed gain phase, block, semi-sectoriality cutoff G} implies that the kernel of $\bm G(s)$ is the kernel of $\bm B^\dagger$, and thus the rank of $\bm G(s)$ is constant.
    Leveraging lemmas from the companion paper \cite{niehues_small-signal_2024}, we have
    \begin{itemize}
        \item $\bm H$ is stable by Lemma 8,
        \item $\bm H$ is frequency-wise sectorial if \eqref{eq:mixed gain phase, block, sectoriality H} holds by Lemma 9,
        \item $\bm G$ is semi-stable by Lemma 10, and
     \item $\bm G$ is frequency-wise semi-sectorial over $(-\omega_c,\omega_c)$ if \eqref{eq:mixed gain phase, block, semi-sectoriality cutoff G} holds by Lemma 11. As the proof of this Lemma goes frequency-wise over $\omega\in\mathbb{R}\setminus\Omega$, it also applies to any interval; this includes $(-\omega_c, \omega_c)$.
    \end{itemize} 
    Using one more time the convex hull property \eqref{eq:conv-hull}, in particular \eqref{eq:phase bounds block matrix}, and the subset property \eqref{eq:sub-field}, the assumptions \eqref{eq:mixed gain phase, block, phase condition sup}-\eqref{eq:mixed gain phase, block, phase condition inf}
imply
\begin{align}
    \sup_{\omega\in[0,\omega_c)\setminus\Omega}\left[\overline{\phi}\left(\bigoplus_n{\bm T}_n\right) + \overline{\phi}\left({\bm B}\bigoplus_l{\bm T}_l{\bm B}^\dagger\right)\right] &< \pi , \\
    \inf_{\omega\in[0,\omega_c)\setminus\Omega}\left[\underline{\phi}\left(\bigoplus_n{\bm T}_n\right) + \underline{\phi}\left({\bm B}
    \bigoplus_l{\bm T}_l{\bm B}^\dagger\right)\right] &> -\pi,
\end{align}
where $\bm T_n$ and $\bm T_l$ are functions of $i \omega$.
These are the phase conditions \eqref{eq:mixed gain phase max phases condition}, \eqref{eq:mixed gain phase min phases condition} of Theorem \ref{thm:mixed gain phase cutoff}.

By \eqref{eq:gain_BB}
and \eqref{eq:gain_blocks}, we have 
that \eqref{eq:mixed gain phase block, gain condition} implies
\eqref{eq:mixed gain phase, gain condition}.
All in all, the system $\left(\bigoplus_n{\bm T}_n\right)\#\left({\bm B}\bigoplus_l{\bm T}_l{\bm B}^\dagger\right)$ then satisfies all assumptions and conditions of Theorem~\ref{thm:mixed gain phase cutoff} and is therefore stable, which concludes the proof.
\end{proof}

\section{Stability of Kuramoto-type adaptive systems}
\label{sec:applications}

We now use Theorem \ref{thm:mixed gain phase cutoff, block}, to derive sufficient criteria for the stability of several Kuramoto-type systems. We demonstrate that this theorem can give novel insights, even for well studied paradigmatic models.

The first example is the classical Kuramoto model. We recover standard results and show that our theorem includes heterogeneous parameters and complex topologies, as well as complexified states and parameters.
The second example is the reformulation of the Kuramoto model with inertia as an adaptive system. Here, we recover established sufficient stability conditions in a novel way. The third system is a proper adaptive Kuramoto model. Here we obtain new sufficient stability conditions that match the necessary conditions of \citet{do_topological_2016}, but also generalize to heterogeneous parameters.

Recall that we denote $[\bm x] = \text{diag}(x_i)$.

\subsection{The classical Kuramoto model and its generalizations}
Our results also apply to the subclass of non-adaptive systems. 
Picking up the motivating example of \ref{sec:high level sketch}, consider the classical Kuramoto model \cite{kuramoto_self-entrainment_1975} of $N$ oscillators with phases $x_n$ and natural frequencies $f_n$, but allow a complex network topology and heterogeneous symmetric coupling weights $K_{nm}$:
\begin{align}
    \dot x_n = f_n - \sum_m K_{nm} \sin(x_n - x_m).
\end{align}
Without loss of generality, we assume $\sum_n f_n = 0$.
Omitting technical details, the linearized system can be represented by the transfer operators $\frac{1}{s} \bm L \# \bm I$, with the Laplacian $\bm L \coloneqq \bm B \bm K \left[\cos\Delta \bm x^\circ\right] \bm B^\dagger$, where $\bm B$ is the node-edge incidence matrix, $\bm K = \text{diag}\left(K_l \right)$ and $\Delta x_l = x_n - x_m$, both for lexicographical edges $l=(n,m)$, $n<m$.
The steady state fulfills $f_n = \sum_m K_{nm} \sin(x_n^\circ - x_m^\circ)$. Applying our method, we find it is stable if $K_l \cos\Delta x_l^\circ > 0$, implying phase differences $\vert\Delta x_l^\circ\vert< \pi/2$ for all edges $l = (n,m)$, a classical result that we get straightforward.

Furthermore, our method almost immediately applies \footnote{To see this, consider the stacked system of complexified $\bm x$ and the complex conjugate $\bm x^*$, which is unitarily equivalent to a real system of stacked $\Re \bm x$ and $\Im \bm x$.} to the complexified Kuramoto model introduced by \citet{thumler_synchrony_2023,lee_complexified_2024}:
The complexified system is stable if $\Re\left(K_l \cos\Delta x_l^\circ\right) > 0$.

\subsection{The Kuramoto model with inertia}
As a first detailed example, we study the Kuramoto model with inertia \cite{bergen_structure_1981,filatrella_analysis_2008} in the most general case of heterogeneous parameters,
\begin{align}\label{eq:Kuramoto_second_order}
    m_n \Ddot{x}_n + \gamma_n \dot{x}_n = P_n  - \sum_m a_{nm} K_{nm} \sin(x_n - x_m),
\end{align}
where $x_n$ represents the phase of the $n$-th oscillator with inertia $m_n>0$, torque $P_n$, and damping $\gamma_n$.
$K_{nm}$ is the coupling strength of the edge connecting $n$ and $m$.
The connectivity structure is given by the $N \times N$ adjacency matrix $\bm A$ with elements $a_{nm} \in \{0, 1\}$.

By introducing the dynamic coupling variable $\kappa_{nm}$, \eqref{eq:Kuramoto_second_order} can be rewritten as a system of $N$ adaptively coupled phase oscillators
\begin{align}
    \label{eq:kuramoto_node}
    m_n \dot x_n + \gamma_n x_n &= \sum_m a_{nm} \kappa_{nm}
    \\
    \label{eq:kuramoto_edge}
    \dot\kappa_{nm} &= P_{nm} - K_{nm} \sin(x_n-x_m),
\end{align}
where $P_{nm} = - P_{mn}$, and $\sum_m P_{nm} = P_n$.
In the stationary, phase-locked state, all oscillators oscillate with a common frequency $\dot x_n^\circ = \Omega_\text{sync} = \sum_n P_n / \sum_n \gamma_n$ and satisfy the steady state equations $P_{nm} = K_{nm} \sin(x_n^\circ-x_m^\circ)$.
Again, we can assume $\sum_n P_n = 0$ without loss of generality.
The $\kappa_{nm}^\circ$ correspond to the excess energy that flowed over the line in order to reach the steady state from an initial condition.

The system can be expressed as a system of coupled node variables $\bm{x} = (x_1, x_2, \dots, x_N)^\top$ and lexicographically ordered edge variables 
$\bm{\kappa} = \left( \kappa_{nm} \mid a_{nm} = 1,\; n<m \right)$.  In this model, the antisymmetry $\kappa_{nm} = -\kappa_{mn}$ holds for the steady state and is conserved by the dynamics.
The $N$ node and $L$ edge variables are linked via an $N \times L$ signed incidence matrix $\bm B$, with all edges oriented from node $n$ to node $m$. The entries of the matrix indicate how the edges connect to the nodes: a value of $+1$ indicates that the edge is directed towards a node, $-1$ indicates the edge is directed away from a node, and $0$ means the node is not connected to that edge. Denote $l = (l_1, l_2) $, then
\begin{align}
B_{nl} &= \delta_{nl_1} - \delta_{nl_2} \; .
\end{align}
In \eqref{eq:kuramoto_node}, we can use $\bm B$ to sum over the edge variables, and in \eqref{eq:kuramoto_edge} we can use $\bm B^\top = \bm B^\dagger$ to represent phase differences. Introducing $\bm M = \text{diag}(m_n)$, $\bm \Gamma = \text{diag}(\gamma_n)$, $\bm K = (\text{diag}(K_{nm}) \mid a_{nm} = 1, n<m)$ and $\bm P = (P_{nm} \mid a_{nm} = 1, n<m)$, the system is written as

\begin{align}
\bm M \dot{\bm x} + \bm \Gamma \bm x &= \bm B \bm \kappa\\
\dot {\bm \kappa} &= \bm P - \bm K \sin(\bm B^\top \bm x)).
\end{align}

To linearize the system around the phase-locked state, we introduce error coordinates $\Delta \bm x  = \bm x -\bm x^\circ$ and $\Delta \bm \kappa  = \bm \kappa - \bm \kappa^\circ$.
The linearized system reads

\begin{align}
\bm M \Delta \dot {\bm x} + \bm \Gamma \Delta \bm x &= \bm B \Delta \bm \kappa\\
\Delta \dot {\bm \kappa} &= \bm K [\cos(\bm B^\top \bm x^\circ)] \bm B^\top \Delta \bm x,
\end{align}
where $[\cos(\bm B^\top \bm x^\circ)]$ is a diagonal matrix.

In Laplace space, the transfer operators of the system  are
\begin{align}
    \Delta {\bm x} &= \bm T^{ve}(s) \bm B \Delta \bm \kappa
    \\
    \bm T^{ve}(s) &= \frac{1}{\bm M s + \bm \Gamma}
    \\
    \Delta {\bm \kappa} &= -\bm T^{ev}(s) \bm B^\top \Delta \bm x
    \\
    \bm T^{ev}(s) &= \frac{1}{s} \bm K[\cos\bm B^\top \bm x^\circ].
\end{align}

With $\bm\chi=\bm B \bm \kappa$, we have
\begin{align}
    \bm \Delta\bm x &= \frac{1}{\bm M s + \bm \Gamma} \bm{\Delta \chi}
    \\
    \bm{\Delta \chi} &= -\bm B \frac{1}{s} \bm K[\cos\bm B^\top \bm x^\circ] \bm B^\top \bm{\Delta x},
\end{align}
and we can write the interconnected feedback loop as $\bm B^\top \bm T^{ve} \bm B \# \bm T^{ev}$.

To apply Proposition \ref{thm:mixed gain phase cutoff, block}, we must ensure that $\bm T^{ve}$ and $\bm T^{ev}$ satisfy the conditions for (semi-)stability and (semi-)sectoriality. Moreover, the DC phase center of both transfer operators is assumed be 0 by convention. 
The diagonal matrices $\bm T^{ve} = \bigoplus_n \bm T^{ve}_n$ and $\bm T^{ev} = \bigoplus_l \bm T^{ev}_l$ have the entries 

\begin{align}
\bm T^{ve}_n(s) &= \frac{1}{m_n s+ \gamma_n} 
\\
\bm T^{ev}_l(s) &= \frac{1}{s} K_{nm} \cos(x_n^\circ - x_m^\circ).
\end{align}

For $\bm T^{ve}(s)$ to be stable, the poles at $s = -\frac{\gamma_n}{m_n}$ must lie in the left half of the complex plane, which requires $\frac{\gamma_n}{m_n}>0$. Since the inertia $m_n>0$, it follows that the damping must be positive: $\gamma_n>0$. It follows that the DC phase center is $0$. Since $\bm T^{ve}(i\omega)$ is a diagonal matrix, its numerical range is the complex hull of the diagonal entries, thus $\bm T^{ve}(s)$ is frequency-wise sectorial.

$\bm T^{ev}(s)$ is marginally stable (or semi-stable) because all poles lie on the imaginary axis at $s=0$. Since $s=0$ is a pole, the DC phase center is evaluated at $s=\epsilon^+$. For $\bm T^{ev}$ to be frequency-wise semi-sectorial and have constant rank along indented imaginary axis with a DC phase center of 0, $\bm T^{ev}(\epsilon^+)$ must be positive definite, which leads to the condition 

\begin{align}
\bm T^{ev}_l(\epsilon^+) &= \frac{1}{\epsilon^+} K_{nm} \cos(x_n^\circ-x_m^\circ) > 0\; ,
\\
|x_n^\circ-x_m^\circ| &< \frac{\pi}{2} \quad\text{if}\quad K_{nm}>0\; ,
\\
|x_n^\circ-x_m^\circ| &> \frac{\pi}{2} \quad\text{if}\quad K_{mn}<0\; ,
\end{align}
for all connected $(n,m)$.

Next, we evaluate the phase and gain conditions of Proposition \ref{thm:mixed gain phase cutoff, block}. The phases of $\bm T^{ve}$ range from $0$ to $-\frac{\pi}{2}$, with $\phi(\bm T^{ve}_n(i\omega)) \in [0,-\frac{\pi}{2}]$ for $\omega \in [0, \infty]$.
The phases of $\bm T^{ev}$ are constant over all finite frequencies, with $\phi(\bm T^{ev}_l(i\omega)) = -\frac{\pi}{2}$ for $\omega \in [0, \infty]\setminus \Omega$.
With $\omega_c \in (0,\infty)$, the phase condition 
\begin{align}
\overline\phi(\bm T^{ve}(i\omega)) + \overline\phi(\bm T^{ev}(i\omega)) < \pi\\
\underline\phi(\bm T^{ve}(i\omega)) + \underline\phi(\bm T^{ev}(i\omega)) > -\pi 
\end{align}
is fulfilled for each $\omega \in [0, \omega_c) \backslash \Omega$.

$\bm T^{ve}$ and $\bm T^{ev}$  have the singular values 
\begin{align}
\sigma(\bm T^{ve}_n(i\omega)) &= \frac{1}{\sqrt{(m_n\omega)^2+\gamma_n^2}}
\\
\sigma(\bm T^{ev}_l(i\omega)) &= \frac{1}{\omega} K_{nm} \cos(x_n^\circ - x_m^\circ), 
\end{align}
where the gains $\overline\sigma(\bm T^{ve}(i\omega))$ and $\overline\sigma(\bm T^{ev}(i\omega))$ correspond to the largest singular values. With $D$ being the maximum degree in the network, $\overline\sigma(\bm B)^2$ is bounded from above by $2D$ (see \eqref{eq:gain_BB}-\eqref{eq:gain_laplacian_bb}).
We can always find a cut-off frequency $\omega_c$ so that the gain condition
\begin{align}
\overline\sigma(\bm B)^2\overline\sigma(\bm T^{ve}) \overline\sigma(\bm T^{ev}) \le 2 D \frac{K_{nm} \cos(x_n^\circ - x_m^\circ)}{\omega \sqrt{(m_n\omega)^2+\gamma_n^2}} < 1 
\end{align}
is fulfilled for each $\omega \in [\omega_c, \infty)$. 

These results show that we can apply Proposition \ref{thm:mixed gain phase cutoff, block} to prove  stability for oscillators with $m_n > 0$, $\gamma_n > 0$, positive coupling $K_{nm} > 0$, and $|x_n^\circ - x_m^\circ| < \frac{\pi}{2}$ for connected $(n,m)$. For oscillators with negative coupling $K_{nm} < 0$, stability is ensured when the phase differences satisfy $|x_n^\circ - x_m^\circ| > \frac{\pi}{2}$ for connected $(n,m)$. 
This is basically the typical Laplacian formulation, but we are allowed to leave out $\bm B$ and $\bm B^\top$ and only analyze the line weights.

\subsection{An adaptive Kuramoto model}
We study the adaptive Kuramoto model presented in \citet{do_topological_2016}, with the modification that $c_{nm}$ can vary for each link. The system consists of $N$ adaptively coupled phase oscillators
\begin{align}
\dot x_n &= \psi_n - \sum_{m \neq n}^N A_{nm} \sin(x_n-x_m)\\
\dot A_{nm} &= \cos(x_n-x_m) - c_{nm} A_{nm},
\end{align}
where $x_n$ denotes the phase and $\psi_n$ the intrinsic frequency of node $n$. The coupling matrix $\mathbf{A} \in \mathbb{R}^{N \times N}$ defines an undirected, weighted network, where two oscillators $n,m$ are connected if $A_{nm} = A_{mn} \neq 0$. The coupling is adaptive and co-evolves with the dynamics of the oscillators. 
In the stationary phase-locked state all oscillators oscillate with a  common frequency $\Psi = \frac{1}{N} \sum_n \psi_n$, which we can assume to be zero without loss of generality, and $A_{nm}^\circ = \cos(x_m^\circ - x_n^\circ)/c_{nm}$.  

The system can be expressed as a system of coupled node variables $\bm{x} = (x_1, x_2, \dots, x_N)^\top$ and edge variables 
$\bm a = (A_{nm} \mid n<m)$, corresponding to the elements from the upper triangular part of the coupling matrix $\bm A$. The $N$ node and $L$ edge variables are linked via the $N \times L$ incidence matrix $\bm B$, where we assume that each edge is oriented from node $n$ to node $m$, with $n<m$ for all edges. Introducing $\bm C = (\text{diag}(c_{nm}) \mid n<m$) and the outputs $\bm o$ and $\bm u$ of the node and edge dynamics, respectively, the system is written as an input-output system matching the form of \eqref{eq:coupled_system_1}-\eqref{eq:coupled_system_4}.
\begin{align}
\dot {\bm x} &= \bm \psi + \bm u \\
\bm o &= \bm x\\
\dot {\bm a} &= \cos(\bm B^\top \bm o) - \bm C\bm a  \\
\bm u &= - \bm B ( \bm a \circ \sin(\bm B^\top \bm o)),
\end{align}
where $\circ$ denotes element-wise multiplication.

To linearize the system around the phase-locked state, we introduce error coordinates $\Delta \bm x  = \bm x -\bm x^\circ$,  $\Delta \bm o  = \bm o -\bm o^\circ$, $\Delta \bm a  = \bm a - \bm a^\circ$ and $\Delta \bm u  = \bm u - \bm u^\circ$, where $\dot {\bm x}^\circ = \Psi$ and $\bm a^\circ = \bm C^{-1} \cos(\bm B^\top \bm x^\circ)$. The linearized system reads
\begin{align}
\Delta \dot {\bm x} &= \Delta \bm u \\
\Delta \bm o &= \Delta \bm x\\
\Delta \dot {\bm a} &= -\bm C \Delta \bm a - [\sin(\bm B^\top \bm x^\circ)] \bm B^\top \Delta \bm o\\
\Delta \bm u &= -\bm B \left([\sin(\bm B^\top \bm x^\circ)] \Delta \bm a + \bm C^{-1} [\cos^2(\bm B^\top \bm x^\circ)] \bm B^\top \Delta \bm o \right).
\end{align}
In Laplace space, the transfer operators of the system  are
\begin{align}
    \Delta {\bm o} &= \bm T^{ve}(s) \Delta \bm u
    \\
    \bm T^{ve}(s) &= \frac{1}{s}
    \\
    \Delta {\bm u} &= -\bm B \bm T^{ev}(s) \bm B^\top \Delta \bm o
    \\
    \bm T^{ev}(s) &= \bm C^{-1} [\cos^2(\bm B^\top \bm x^\circ)] -  (s[\bm 1] +\bm C)^{-1} [\sin^2(\bm B^\top \bm x^\circ)].
\end{align}

To simplify the analysis, we introduce $\Delta \bm v = \frac{1}{s} \Delta \bm u$,
\begin{align}
    \Delta {\bm o} &= \bm T^{ve}(s) \Delta \bm v
    \\
    \bm T^{ve}(s) &= \bm I
    \\
    \Delta {\bm v} &= -\bm B \bm T^{ev}(s) \bm B^\top \Delta \bm o
    \\
    \bm T^{ev}(s) &= (s\bm C)^{-1} [\cos^2(\bm B^\top \bm x^\circ)] \nonumber\\&\quad -  (s(s[\bm 1] +\bm C))^{-1} [\sin^2(\bm B^\top \bm x^\circ)].
\end{align}

To apply Proposition \ref{thm:mixed gain phase cutoff, block}, we must ensure that the transfer operators satisfy the conditions for (semi-)stability, (semi-)sectoriality and DC phase center.
$\bm T^{ve} = \bm I$ is stable and frequency-wise semi-sectorial with a DC phase center of $\bm T^{ve}(0) =0$. 
The diagonal matrix $\bm{T^{ev}} = \bigoplus_l \bm{T^{ev}}_l$ has the entries 
\begin{align}
\bm T^{ev}_l(s) =  \frac{1}{s c_{nm}} \cos^2(x_n^\circ - x_m^\circ) -  \frac{1}{s(s+c_{nm})} \sin^2(x_n^\circ - x_m^\circ)
\end{align}
$\bm T^{ev}(s)$ has poles at $s=0$ and at $s = -c_{nm}$, which requires $c_{nm}\ge0$ for 
semi-stability and $c_{nm}\neq0$ for boundedness. For $\bm T^{ev}$ to be frequency-wise semi-sectorial and have full rank along the indented imaginary axis with a DC phase center of 0, $\bm T^{ev}(\epsilon^+)$ must be positive definite, which leads to the condition

\begin{align}
\frac{1}{ c_{nm}} \cos^2(x_n^\circ - x_m^\circ) -  \frac{1}{\epsilon^+ + c_{nm}} \sin^2(x_n^\circ - x_m^\circ) &> 0.
\end{align}
This must hold for arbitrarily small $\epsilon^+$, thus
\begin{align}
\frac{1}{ c_{nm}} \cos^2(x_n^\circ - x_m^\circ) -  \frac{1}{ c_{nm}} \sin^2(x_n^\circ - x_m^\circ) &> 0\\
\cos(2(x_n^\circ-x_m^\circ)) &> 0\\
|x_n^\circ-x_m^\circ| &> \frac{\pi}{4}.
\end{align}
Next, we analyze the interconnected system $\bm T^{ve} \# \bm B \bm T^{ev} \bm B^\top$ and evaluate the phase and gain conditions of Proposition \ref{thm:mixed gain phase cutoff, block}.

$\bm T^{ve}$ has phase $\phi(\bm T^{ve}(i\omega)) = 0$ for all $s$.
The phases of $\bm T^{ev}$ lie in $[0,\phi_c] - \frac{\pi}{2}$ for $\omega \in [0, \infty]\setminus \Omega$, with $\phi_c \in \left[\arctan\frac12,\frac{\pi}{4}\right]$ depending on the phase differences:
\begin{align}
    \phi_c \coloneqq \max_{l=(n,m)} \arctan\left\{ \left[2 \cos^2(x_n^\circ - x_m^\circ) \right]^{-1} \right\}.
\end{align}
Hence, the phase condition 
\begin{align}
\overline\phi(\bm T^{ve}(i\omega)) + \overline\phi(\bm T^{ev}(i\omega)) &< \pi\\
\underline\phi(\bm T^{ve}(i\omega)) + \underline\phi(\bm T^{ev}(i\omega)) &> -\pi 
\end{align}
is fulfilled for each $\omega \in [0, \infty]\backslash \Omega$.
We therefore do not need the gain condition.

We can apply Proposition \ref{thm:mixed gain phase cutoff, block} and prove stability for $c_{nm} > 0$ and oscillators satisfying $|x_n^\circ-x_m^\circ| < \frac{\pi}{4}$ for all edges. This agrees with the necessary stability conditions presented in \citet{do_topological_2016} for the case of homogeneous parameters, showing that the condition is exact in this setting. Thus, the combined results completely settle the linear stability of this adaptive Kuramoto model.

\section{Discussion}
This paper introduced a new method for studying the linear stability of steady states in adaptive networks. By leveraging new results from linear algebra and control theory, we could give novel stability conditions. Applying these to the classical Kuramoto model and the Kuramoto model with inertia (written as adaptive first-order oscillators) demonstrated that they are not overly conservative, and allow straightforward generalization to heterogeneous parameters in a local fashion. In the case of truly adaptive oscillators, the sufficient conditions we obtain even match necessary conditions derived previously. However, our conditions apply much more broadly, including in the case of heterogeneous parameters as well. This gives a complete characterization of the stability of steady states of this adaptive Kuramoto model in situations in which the necessary result applies.

One notable aspect of the theory developed here, is that nodes and edges are allowed to be highly heterogeneous in their states and dynamics. Contrary to work in the master stability tradition, we do not require any factorization of the systems Jacobian of the type $\bm J = \bm F \otimes \bm I + \bm G \otimes \bm L$. This could be particularly important for the study of multilayer networks with different topologies in the different layers, where no such factorization is available.

The natural way in which phase conditions and network structure interact, also allows many further generalizations of the above framework. Due to the fact that edge states have their own full dynamics, there is considerable freedom in choosing exactly how the system is partitioned into edges and nodes. The structured perturbation approach of \citet{woolcock_mixed_2023} is an example of this.

There also are further variations of phase-based stability conditions that could be interesting to study, particularly for directed networks. Notions like $r$-sectoriality can combine gain and phase information in ways that encode directedness in a natural way \cite{wang_phases_2023}.

Finally, the examples explored above used a signed incidence matrix, and thus Laplacian-type diffusive coupling. However, the theorem leaves the nature of $\bm B$ completely open. We leave exploring the use of this theorem for systems where the coupling is not Laplacian in nature, e.g., additive, to future work.

\section{Data availability statement}
The data that supports the findings of this study are available within the article.

\section{Author Declarations}
The authors have no conflicts to disclose

\section{Acknowledgements}

This work was supported by the OpPoDyn Project, Federal Ministry for Economic Affairs and Climate Action (FKZ:03EI1071A).

J.N. gratefully acknowledges support by \mbox{BIMoS} (TU Berlin), Studienstiftung des Deutschen Volkes, and the Berlin Mathematical School, funded by the Deutsche Forschungsgemeinschaft (DFG, German Research Foundation) Germany's Excellence Strategy --- The Berlin Mathematics Research Center MATH+ (EXC-2046/1, project ID: 390685689).

R.D. was supported by the Swiss
National Science Foundation under grant
nr. 200021\_215336.

\bibliography{zotero,somemore}

\begin{thebibliography}{40}%
\makeatletter
\providecommand \@ifxundefined [1]{%
 \@ifx{#1\undefined}
}%
\providecommand \@ifnum [1]{%
 \ifnum #1\expandafter \@firstoftwo
 \else \expandafter \@secondoftwo
 \fi
}%
\providecommand \@ifx [1]{%
 \ifx #1\expandafter \@firstoftwo
 \else \expandafter \@secondoftwo
 \fi
}%
\providecommand \natexlab [1]{#1}%
\providecommand \enquote  [1]{``#1''}%
\providecommand \bibnamefont  [1]{#1}%
\providecommand \bibfnamefont [1]{#1}%
\providecommand \citenamefont [1]{#1}%
\providecommand \href@noop [0]{\@secondoftwo}%
\providecommand \href [0]{\begingroup \@sanitize@url \@href}%
\providecommand \@href[1]{\@@startlink{#1}\@@href}%
\providecommand \@@href[1]{\endgroup#1\@@endlink}%
\providecommand \@sanitize@url [0]{\catcode `\\12\catcode `\$12\catcode `\&12\catcode `\#12\catcode `\^12\catcode `\_12\catcode `\%12\relax}%
\providecommand \@@startlink[1]{}%
\providecommand \@@endlink[0]{}%
\providecommand \url  [0]{\begingroup\@sanitize@url \@url }%
\providecommand \@url [1]{\endgroup\@href {#1}{\urlprefix }}%
\providecommand \urlprefix  [0]{URL }%
\providecommand \Eprint [0]{\href }%
\providecommand \doibase [0]{http://dx.doi.org/}%
\providecommand \selectlanguage [0]{\@gobble}%
\providecommand \bibinfo  [0]{\@secondoftwo}%
\providecommand \bibfield  [0]{\@secondoftwo}%
\providecommand \translation [1]{[#1]}%
\providecommand \BibitemOpen [0]{}%
\providecommand \bibitemStop [0]{}%
\providecommand \bibitemNoStop [0]{.\EOS\space}%
\providecommand \EOS [0]{\spacefactor3000\relax}%
\providecommand \BibitemShut  [1]{\csname bibitem#1\endcsname}%
\let\auto@bib@innerbib\@empty
\bibitem [{\citenamefont {Gross}\ and\ \citenamefont {Blasius}(2007)}]{anreview}%
  \BibitemOpen
  \bibfield  {author} {\bibinfo {author} {\bibfnamefont {T.}~\bibnamefont {Gross}}\ and\ \bibinfo {author} {\bibfnamefont {B.}~\bibnamefont {Blasius}},\ }\bibfield  {title} {\enquote {\bibinfo {title} {Adaptive coevolutionary networks: a review},}\ }\href {\doibase 10.1098/rsif.2007.1229} {\bibfield  {journal} {\bibinfo  {journal} {Journal of The Royal Society Interface}\ }\textbf {\bibinfo {volume} {5}},\ \bibinfo {pages} {259--271} (\bibinfo {year} {2007})}\BibitemShut {NoStop}%
\bibitem [{\citenamefont {Gross}, \citenamefont {Dommar~D{'}Lima},\ and\ \citenamefont {Blasius}(2006)}]{anepi}%
  \BibitemOpen
  \bibfield  {author} {\bibinfo {author} {\bibfnamefont {T.}~\bibnamefont {Gross}}, \bibinfo {author} {\bibfnamefont {C.~J.}\ \bibnamefont {Dommar~D{'}Lima}}, \ and\ \bibinfo {author} {\bibfnamefont {B.}~\bibnamefont {Blasius}},\ }\bibfield  {title} {\enquote {\bibinfo {title} {Epidemic dynamics on an adaptive network},}\ }\href {\doibase 10.1103/PhysRevLett.96.208701} {\bibfield  {journal} {\bibinfo  {journal} {Physical Review Letters}\ }\textbf {\bibinfo {volume} {96}},\ \bibinfo {pages} {208701} (\bibinfo {year} {2006})}\BibitemShut {NoStop}%
\bibitem [{\citenamefont {Marceau}\ \emph {et~al.}(2010)\citenamefont {Marceau}, \citenamefont {No{\"e}l}, \citenamefont {H{\'e}bert-Dufresne}, \citenamefont {Allard},\ and\ \citenamefont {Dub{\'e}}}]{marceau}%
  \BibitemOpen
  \bibfield  {author} {\bibinfo {author} {\bibfnamefont {V.}~\bibnamefont {Marceau}}, \bibinfo {author} {\bibfnamefont {P.-A.}\ \bibnamefont {No{\"e}l}}, \bibinfo {author} {\bibfnamefont {L.}~\bibnamefont {H{\'e}bert-Dufresne}}, \bibinfo {author} {\bibfnamefont {A.}~\bibnamefont {Allard}}, \ and\ \bibinfo {author} {\bibfnamefont {L.~J.}\ \bibnamefont {Dub{\'e}}},\ }\bibfield  {title} {\enquote {\bibinfo {title} {Adaptive networks: Coevolution of disease and topology},}\ }\href@noop {} {\bibfield  {journal} {\bibinfo  {journal} {Physical Review E—Statistical, Nonlinear, and Soft Matter Physics}\ }\textbf {\bibinfo {volume} {82}},\ \bibinfo {pages} {036116} (\bibinfo {year} {2010})}\BibitemShut {NoStop}%
\bibitem [{\citenamefont {Scarpino}, \citenamefont {Allard},\ and\ \citenamefont {H{\'e}bert-Dufresne}(2016)}]{scarpino}%
  \BibitemOpen
  \bibfield  {author} {\bibinfo {author} {\bibfnamefont {S.~V.}\ \bibnamefont {Scarpino}}, \bibinfo {author} {\bibfnamefont {A.}~\bibnamefont {Allard}}, \ and\ \bibinfo {author} {\bibfnamefont {L.}~\bibnamefont {H{\'e}bert-Dufresne}},\ }\bibfield  {title} {\enquote {\bibinfo {title} {The effect of a prudent adaptive behaviour on disease transmission},}\ }\href@noop {} {\bibfield  {journal} {\bibinfo  {journal} {Nature Physics}\ }\textbf {\bibinfo {volume} {12}},\ \bibinfo {pages} {1042--1046} (\bibinfo {year} {2016})}\BibitemShut {NoStop}%
\bibitem [{\citenamefont {Kozma}\ and\ \citenamefont {Barrat}(2008)}]{kozma}%
  \BibitemOpen
  \bibfield  {author} {\bibinfo {author} {\bibfnamefont {B.}~\bibnamefont {Kozma}}\ and\ \bibinfo {author} {\bibfnamefont {A.}~\bibnamefont {Barrat}},\ }\bibfield  {title} {\enquote {\bibinfo {title} {Consensus formation on adaptive networks},}\ }\href {\doibase 10.1103/PhysRevE.77.016102} {\bibfield  {journal} {\bibinfo  {journal} {Phys. Rev. E}\ }\textbf {\bibinfo {volume} {77}},\ \bibinfo {pages} {016102} (\bibinfo {year} {2008})}\BibitemShut {NoStop}%
\bibitem [{\citenamefont {Rainer}\ and\ \citenamefont {Krause}(2002)}]{Hegselmann}%
  \BibitemOpen
  \bibfield  {author} {\bibinfo {author} {\bibfnamefont {H.}~\bibnamefont {Rainer}}\ and\ \bibinfo {author} {\bibfnamefont {U.}~\bibnamefont {Krause}},\ }\bibfield  {title} {\enquote {\bibinfo {title} {Opinion dynamics and bounded confidence: models, analysis and simulation},}\ }\href@noop {} {\bibfield  {journal} {\bibinfo  {journal} {JASSS}\ }\textbf {\bibinfo {volume} {5}},\ \bibinfo {pages} {1--12} (\bibinfo {year} {2002})}\BibitemShut {NoStop}%
\bibitem [{\citenamefont {Vazquez}, \citenamefont {Egu{\'\i}luz},\ and\ \citenamefont {Miguel}(2008)}]{sanmiguel}%
  \BibitemOpen
  \bibfield  {author} {\bibinfo {author} {\bibfnamefont {F.}~\bibnamefont {Vazquez}}, \bibinfo {author} {\bibfnamefont {V.~M.}\ \bibnamefont {Egu{\'\i}luz}}, \ and\ \bibinfo {author} {\bibfnamefont {M.~S.}\ \bibnamefont {Miguel}},\ }\bibfield  {title} {\enquote {\bibinfo {title} {Generic absorbing transition in coevolution dynamics},}\ }\href@noop {} {\bibfield  {journal} {\bibinfo  {journal} {Physical review letters}\ }\textbf {\bibinfo {volume} {100}},\ \bibinfo {pages} {108702} (\bibinfo {year} {2008})}\BibitemShut {NoStop}%
\bibitem [{\citenamefont {Couzin}\ \emph {et~al.}(2011)\citenamefont {Couzin}, \citenamefont {Ioannou}, \citenamefont {Demirel}, \citenamefont {Gross}, \citenamefont {Torney}, \citenamefont {Hartnett}, \citenamefont {Conradt}, \citenamefont {Levin},\ and\ \citenamefont {Leonard}}]{couzin}%
  \BibitemOpen
  \bibfield  {author} {\bibinfo {author} {\bibfnamefont {I.~D.}\ \bibnamefont {Couzin}}, \bibinfo {author} {\bibfnamefont {C.~C.}\ \bibnamefont {Ioannou}}, \bibinfo {author} {\bibfnamefont {G.}~\bibnamefont {Demirel}}, \bibinfo {author} {\bibfnamefont {T.}~\bibnamefont {Gross}}, \bibinfo {author} {\bibfnamefont {C.~J.}\ \bibnamefont {Torney}}, \bibinfo {author} {\bibfnamefont {A.}~\bibnamefont {Hartnett}}, \bibinfo {author} {\bibfnamefont {L.}~\bibnamefont {Conradt}}, \bibinfo {author} {\bibfnamefont {S.~A.}\ \bibnamefont {Levin}}, \ and\ \bibinfo {author} {\bibfnamefont {N.~E.}\ \bibnamefont {Leonard}},\ }\bibfield  {title} {\enquote {\bibinfo {title} {Uninformed individuals promote democratic consensus in animal groups},}\ }\href {\doibase 10.1126/science.1210280} {\bibfield  {journal} {\bibinfo  {journal} {Science}\ }\textbf {\bibinfo {volume} {334}},\ \bibinfo {pages} {1578--1580} (\bibinfo {year} {2011})}\BibitemShut {NoStop}%
\bibitem [{\citenamefont {Skyrms}\ and\ \citenamefont {Pemantle}(2009)}]{skyrms}%
  \BibitemOpen
  \bibfield  {author} {\bibinfo {author} {\bibfnamefont {B.}~\bibnamefont {Skyrms}}\ and\ \bibinfo {author} {\bibfnamefont {R.}~\bibnamefont {Pemantle}},\ }\bibfield  {title} {\enquote {\bibinfo {title} {A dynamic model of social network formation},}\ }in\ \href@noop {} {\emph {\bibinfo {booktitle} {Adaptive Networks: Theory, Models and Applications}}}\ (\bibinfo  {publisher} {Springer},\ \bibinfo {year} {2009})\ pp.\ \bibinfo {pages} {231--251}\BibitemShut {NoStop}%
\bibitem [{\citenamefont {Do}, \citenamefont {Rudolf},\ and\ \citenamefont {Gross}(2010)}]{do}%
  \BibitemOpen
  \bibfield  {author} {\bibinfo {author} {\bibfnamefont {A.-L.}\ \bibnamefont {Do}}, \bibinfo {author} {\bibfnamefont {L.}~\bibnamefont {Rudolf}}, \ and\ \bibinfo {author} {\bibfnamefont {T.}~\bibnamefont {Gross}},\ }\bibfield  {title} {\enquote {\bibinfo {title} {Patterns of cooperation: fairness and coordination in networks of interacting agents},}\ }\href {\doibase 10.1088/1367-2630/12/6/063023} {\bibfield  {journal} {\bibinfo  {journal} {New Journal of Physics}\ }\textbf {\bibinfo {volume} {12}},\ \bibinfo {pages} {063023} (\bibinfo {year} {2010})}\BibitemShut {NoStop}%
\bibitem [{\citenamefont {Bornholdt}\ and\ \citenamefont {Rohlf}(2000)}]{bornholdt}%
  \BibitemOpen
  \bibfield  {author} {\bibinfo {author} {\bibfnamefont {S.}~\bibnamefont {Bornholdt}}\ and\ \bibinfo {author} {\bibfnamefont {T.}~\bibnamefont {Rohlf}},\ }\bibfield  {title} {\enquote {\bibinfo {title} {Topological evolution of dynamical networks: Global criticality from local dynamics},}\ }\href@noop {} {\bibfield  {journal} {\bibinfo  {journal} {Physical Review Letters}\ }\textbf {\bibinfo {volume} {84}},\ \bibinfo {pages} {6114} (\bibinfo {year} {2000})}\BibitemShut {NoStop}%
\bibitem [{\citenamefont {Meisel}\ and\ \citenamefont {Gross}(2009)}]{meisel}%
  \BibitemOpen
  \bibfield  {author} {\bibinfo {author} {\bibfnamefont {C.}~\bibnamefont {Meisel}}\ and\ \bibinfo {author} {\bibfnamefont {T.}~\bibnamefont {Gross}},\ }\bibfield  {title} {\enquote {\bibinfo {title} {Adaptive self-organization in a realistic neural network model},}\ }\href@noop {} {\bibfield  {journal} {\bibinfo  {journal} {Physical Review E—Statistical, Nonlinear, and Soft Matter Physics}\ }\textbf {\bibinfo {volume} {80}},\ \bibinfo {pages} {061917} (\bibinfo {year} {2009})}\BibitemShut {NoStop}%
\bibitem [{\citenamefont {Kuehn}(2012)}]{kuehn}%
  \BibitemOpen
  \bibfield  {author} {\bibinfo {author} {\bibfnamefont {C.}~\bibnamefont {Kuehn}},\ }\bibfield  {title} {\enquote {\bibinfo {title} {Time-scale and noise optimality in self-organized critical adaptive networks},}\ }\href@noop {} {\bibfield  {journal} {\bibinfo  {journal} {Physical Review E}\ }\textbf {\bibinfo {volume} {85}},\ \bibinfo {pages} {026103} (\bibinfo {year} {2012})}\BibitemShut {NoStop}%
\bibitem [{\citenamefont {Raimundo}, \citenamefont {Guimaraes},\ and\ \citenamefont {Evans}(2018)}]{ecorev}%
  \BibitemOpen
  \bibfield  {author} {\bibinfo {author} {\bibfnamefont {R.~L.~G.}\ \bibnamefont {Raimundo}}, \bibinfo {author} {\bibfnamefont {P.~R.}\ \bibnamefont {Guimaraes}}, \ and\ \bibinfo {author} {\bibfnamefont {D.~M.}\ \bibnamefont {Evans}},\ }\bibfield  {title} {\enquote {\bibinfo {title} {Adaptive networks for restoration ecology},}\ }\href@noop {} {\bibfield  {journal} {\bibinfo  {journal} {Trends in Ecology and Evolution}\ }\textbf {\bibinfo {volume} {33}},\ \bibinfo {pages} {664--675} (\bibinfo {year} {2018})}\BibitemShut {NoStop}%
\bibitem [{\citenamefont {Berner}\ \emph {et~al.}(2023)\citenamefont {Berner}, \citenamefont {Gross}, \citenamefont {Kuehn}, \citenamefont {Kurths},\ and\ \citenamefont {Yanchuk}}]{bernerrev}%
  \BibitemOpen
  \bibfield  {author} {\bibinfo {author} {\bibfnamefont {R.}~\bibnamefont {Berner}}, \bibinfo {author} {\bibfnamefont {T.}~\bibnamefont {Gross}}, \bibinfo {author} {\bibfnamefont {C.}~\bibnamefont {Kuehn}}, \bibinfo {author} {\bibfnamefont {J.}~\bibnamefont {Kurths}}, \ and\ \bibinfo {author} {\bibfnamefont {S.}~\bibnamefont {Yanchuk}},\ }\bibfield  {title} {\enquote {\bibinfo {title} {Adaptive dynamical networks},}\ }\href@noop {} {\bibfield  {journal} {\bibinfo  {journal} {Physics Reports}\ }\textbf {\bibinfo {volume} {1031}},\ \bibinfo {pages} {1--59} (\bibinfo {year} {2023})}\BibitemShut {NoStop}%
\bibitem [{\citenamefont {Demirel}\ \emph {et~al.}(2014)\citenamefont {Demirel}, \citenamefont {Vazquez}, \citenamefont {B{\"o}hme},\ and\ \citenamefont {Gross}}]{demirel}%
  \BibitemOpen
  \bibfield  {author} {\bibinfo {author} {\bibfnamefont {G.}~\bibnamefont {Demirel}}, \bibinfo {author} {\bibfnamefont {F.}~\bibnamefont {Vazquez}}, \bibinfo {author} {\bibfnamefont {G.~A.}\ \bibnamefont {B{\"o}hme}}, \ and\ \bibinfo {author} {\bibfnamefont {T.}~\bibnamefont {Gross}},\ }\bibfield  {title} {\enquote {\bibinfo {title} {Moment-closure approximations for discrete adaptive networks},}\ }\href@noop {} {\bibfield  {journal} {\bibinfo  {journal} {Physica D}\ }\textbf {\bibinfo {volume} {267}},\ \bibinfo {pages} {68--80} (\bibinfo {year} {2014})}\BibitemShut {NoStop}%
\bibitem [{\citenamefont {Segel}\ and\ \citenamefont {Levin}(1976)}]{levin}%
  \BibitemOpen
  \bibfield  {author} {\bibinfo {author} {\bibfnamefont {L.~A.}\ \bibnamefont {Segel}}\ and\ \bibinfo {author} {\bibfnamefont {S.~A.}\ \bibnamefont {Levin}},\ }\bibfield  {title} {\enquote {\bibinfo {title} {Application of nonlinear stability theory to the study of the effects of diffusion on predator-prey interactions},}\ }in\ \href@noop {} {\emph {\bibinfo {booktitle} {AIP conference proceedings}}},\ Vol.~\bibinfo {volume} {27}\ (\bibinfo {organization} {American Institute of Physics},\ \bibinfo {year} {1976})\ pp.\ \bibinfo {pages} {123--152}\BibitemShut {NoStop}%
\bibitem [{\citenamefont {Pecora}\ and\ \citenamefont {Carroll}(1998)}]{pecora_master_1998}%
  \BibitemOpen
  \bibfield  {author} {\bibinfo {author} {\bibfnamefont {L.~M.}\ \bibnamefont {Pecora}}\ and\ \bibinfo {author} {\bibfnamefont {T.~L.}\ \bibnamefont {Carroll}},\ }\bibfield  {title} {\enquote {\bibinfo {title} {Master stability functions for synchronized coupled systems},}\ }\href@noop {} {\bibfield  {journal} {\bibinfo  {journal} {Physical review letters}\ }\textbf {\bibinfo {volume} {80}},\ \bibinfo {pages} {2109} (\bibinfo {year} {1998})},\ \bibinfo {note} {publisher: APS}\BibitemShut {NoStop}%
\bibitem [{\citenamefont {Berner}\ \emph {et~al.}(2021)\citenamefont {Berner}, \citenamefont {Vock}, \citenamefont {Schöll},\ and\ \citenamefont {Yanchuk}}]{berner_desynchronization_2021}%
  \BibitemOpen
  \bibfield  {author} {\bibinfo {author} {\bibfnamefont {R.}~\bibnamefont {Berner}}, \bibinfo {author} {\bibfnamefont {S.}~\bibnamefont {Vock}}, \bibinfo {author} {\bibfnamefont {E.}~\bibnamefont {Schöll}}, \ and\ \bibinfo {author} {\bibfnamefont {S.}~\bibnamefont {Yanchuk}},\ }\bibfield  {title} {\enquote {\bibinfo {title} {Desynchronization transitions in adaptive networks},}\ }\href@noop {} {\bibfield  {journal} {\bibinfo  {journal} {Physical Review Letters}\ }\textbf {\bibinfo {volume} {126}},\ \bibinfo {pages} {028301} (\bibinfo {year} {2021})},\ \bibinfo {note} {publisher: APS}\BibitemShut {NoStop}%
\bibitem [{\citenamefont {Berner}\ and\ \citenamefont {Yanchuk}(2021)}]{berner_synchronization_2021}%
  \BibitemOpen
  \bibfield  {author} {\bibinfo {author} {\bibfnamefont {R.}~\bibnamefont {Berner}}\ and\ \bibinfo {author} {\bibfnamefont {S.}~\bibnamefont {Yanchuk}},\ }\bibfield  {title} {\enquote {\bibinfo {title} {Synchronization in networks with heterogeneous adaptation rules and applications to distance-dependent synaptic plasticity},}\ }\href@noop {} {\bibfield  {journal} {\bibinfo  {journal} {Frontiers in Applied Mathematics and Statistics}\ }\textbf {\bibinfo {volume} {7}},\ \bibinfo {pages} {714978} (\bibinfo {year} {2021})},\ \bibinfo {note} {publisher: Frontiers Media SA}\BibitemShut {NoStop}%
\bibitem [{\citenamefont {Wang}\ \emph {et~al.}(2020)\citenamefont {Wang}, \citenamefont {Chen}, \citenamefont {Khong},\ and\ \citenamefont {Qiu}}]{wang_phases_2020}%
  \BibitemOpen
  \bibfield  {author} {\bibinfo {author} {\bibfnamefont {D.}~\bibnamefont {Wang}}, \bibinfo {author} {\bibfnamefont {W.}~\bibnamefont {Chen}}, \bibinfo {author} {\bibfnamefont {S.~Z.}\ \bibnamefont {Khong}}, \ and\ \bibinfo {author} {\bibfnamefont {L.}~\bibnamefont {Qiu}},\ }\bibfield  {title} {\enquote {\bibinfo {title} {On the phases of a complex matrix},}\ }\href {\doibase 10.1016/j.laa.2020.01.035} {\bibfield  {journal} {\bibinfo  {journal} {Linear Algebra and its Applications}\ }\textbf {\bibinfo {volume} {593}},\ \bibinfo {pages} {152--179} (\bibinfo {year} {2020})}\BibitemShut {NoStop}%
\bibitem [{\citenamefont {Wang}\ \emph {et~al.}(2023)\citenamefont {Wang}, \citenamefont {Mao}, \citenamefont {Chen},\ and\ \citenamefont {Qiu}}]{wang_phases_2023}%
  \BibitemOpen
  \bibfield  {author} {\bibinfo {author} {\bibfnamefont {D.}~\bibnamefont {Wang}}, \bibinfo {author} {\bibfnamefont {X.}~\bibnamefont {Mao}}, \bibinfo {author} {\bibfnamefont {W.}~\bibnamefont {Chen}}, \ and\ \bibinfo {author} {\bibfnamefont {L.}~\bibnamefont {Qiu}},\ }\bibfield  {title} {\enquote {\bibinfo {title} {On the phases of a semi-sectorial matrix and the essential phase of a {Laplacian}},}\ }\href {\doibase 10.1016/j.laa.2023.07.014} {\bibfield  {journal} {\bibinfo  {journal} {Linear Algebra and its Applications}\ }\textbf {\bibinfo {volume} {676}},\ \bibinfo {pages} {441--458} (\bibinfo {year} {2023})}\BibitemShut {NoStop}%
\bibitem [{\citenamefont {Chen}\ \emph {et~al.}(2024)\citenamefont {Chen}, \citenamefont {Wang}, \citenamefont {Khong},\ and\ \citenamefont {Qiu}}]{chen_phase_2024}%
  \BibitemOpen
  \bibfield  {author} {\bibinfo {author} {\bibfnamefont {W.}~\bibnamefont {Chen}}, \bibinfo {author} {\bibfnamefont {D.}~\bibnamefont {Wang}}, \bibinfo {author} {\bibfnamefont {S.~Z.}\ \bibnamefont {Khong}}, \ and\ \bibinfo {author} {\bibfnamefont {L.}~\bibnamefont {Qiu}},\ }\bibfield  {title} {{\selectlanguage {en}\enquote {\bibinfo {title} {A {Phase} {Theory} of {Multi}-{Input} {Multi}-{Output} {Linear} {Time}-{Invariant} {Systems}},}\ }}\href {\doibase 10.1137/22M148968X} {\bibfield  {journal} {\bibinfo  {journal} {SIAM Journal on Control and Optimization}\ }\textbf {\bibinfo {volume} {62}},\ \bibinfo {pages} {1235--1260} (\bibinfo {year} {2024})}\BibitemShut {NoStop}%
\bibitem [{\citenamefont {Zhao}, \citenamefont {Chen},\ and\ \citenamefont {Qiu}(2022)}]{zhao_when_2022}%
  \BibitemOpen
  \bibfield  {author} {\bibinfo {author} {\bibfnamefont {D.}~\bibnamefont {Zhao}}, \bibinfo {author} {\bibfnamefont {W.}~\bibnamefont {Chen}}, \ and\ \bibinfo {author} {\bibfnamefont {L.}~\bibnamefont {Qiu}},\ }\href {\doibase 10.48550/arXiv.2201.06041} {\enquote {\bibinfo {title} {When {Small} {Gain} {Meets} {Small} {Phase}},}\ } (\bibinfo {year} {2022}),\ \bibinfo {note} {arXiv:2201.06041 [cs, eess] version: 2}\BibitemShut {NoStop}%
\bibitem [{\citenamefont {Do}\ \emph {et~al.}(2016)\citenamefont {Do}, \citenamefont {Boccaletti}, \citenamefont {Epperlein}, \citenamefont {Siegmund},\ and\ \citenamefont {Gross}}]{do_topological_2016}%
  \BibitemOpen
  \bibfield  {author} {\bibinfo {author} {\bibfnamefont {A.~L.}\ \bibnamefont {Do}}, \bibinfo {author} {\bibfnamefont {S.}~\bibnamefont {Boccaletti}}, \bibinfo {author} {\bibfnamefont {J.}~\bibnamefont {Epperlein}}, \bibinfo {author} {\bibfnamefont {S.}~\bibnamefont {Siegmund}}, \ and\ \bibinfo {author} {\bibfnamefont {T.}~\bibnamefont {Gross}},\ }\bibfield  {title} {{\selectlanguage {en}\enquote {\bibinfo {title} {Topological stability criteria for networking dynamical systems with {Hermitian} {Jacobian}},}\ }}\href {\doibase 10.1017/S0956792516000425} {\bibfield  {journal} {\bibinfo  {journal} {European Journal of Applied Mathematics}\ }\textbf {\bibinfo {volume} {27}},\ \bibinfo {pages} {888--903} (\bibinfo {year} {2016})}\BibitemShut {NoStop}%
\bibitem [{\citenamefont {Niehues}, \citenamefont {Delabays},\ and\ \citenamefont {Hellmann}(2024)}]{niehues_small-signal_2024}%
  \BibitemOpen
  \bibfield  {author} {\bibinfo {author} {\bibfnamefont {J.}~\bibnamefont {Niehues}}, \bibinfo {author} {\bibfnamefont {R.}~\bibnamefont {Delabays}}, \ and\ \bibinfo {author} {\bibfnamefont {F.}~\bibnamefont {Hellmann}},\ }\href {\doibase 10.48550/arXiv.2411.10832} {\enquote {\bibinfo {title} {Small-signal stability of power systems with voltage droop},}\ } (\bibinfo {year} {2024}),\ \bibinfo {note} {arXiv:2411.10832}\BibitemShut {NoStop}%
\bibitem [{\citenamefont {Kogler}\ \emph {et~al.}(2022)\citenamefont {Kogler}, \citenamefont {Plietzsch}, \citenamefont {Schultz},\ and\ \citenamefont {Hellmann}}]{kogler_normal_2022}%
  \BibitemOpen
  \bibfield  {author} {\bibinfo {author} {\bibfnamefont {R.}~\bibnamefont {Kogler}}, \bibinfo {author} {\bibfnamefont {A.}~\bibnamefont {Plietzsch}}, \bibinfo {author} {\bibfnamefont {P.}~\bibnamefont {Schultz}}, \ and\ \bibinfo {author} {\bibfnamefont {F.}~\bibnamefont {Hellmann}},\ }\bibfield  {title} {\enquote {\bibinfo {title} {Normal {Form} for {Grid}-{Forming} {Power} {Grid} {Actors}},}\ }\href {\doibase 10.1103/PRXEnergy.1.013008} {\bibfield  {journal} {\bibinfo  {journal} {PRX Energy}\ }\textbf {\bibinfo {volume} {1}},\ \bibinfo {pages} {013008} (\bibinfo {year} {2022})},\ \bibinfo {note} {publisher: American Physical Society}\BibitemShut {NoStop}%
\bibitem [{\citenamefont {Büttner}\ and\ \citenamefont {Hellmann}(2024)}]{buttner_complex_2024}%
  \BibitemOpen
  \bibfield  {author} {\bibinfo {author} {\bibfnamefont {A.}~\bibnamefont {Büttner}}\ and\ \bibinfo {author} {\bibfnamefont {F.}~\bibnamefont {Hellmann}},\ }\bibfield  {title} {\enquote {\bibinfo {title} {Complex {Couplings} - {A} {Universal}, {Adaptive}, and {Bilinear} {Formulation} of {Power} {Grid} {Dynamics}},}\ }\href {\doibase 10.1103/PRXEnergy.3.013005} {\bibfield  {journal} {\bibinfo  {journal} {PRX Energy}\ }\textbf {\bibinfo {volume} {3}},\ \bibinfo {pages} {013005} (\bibinfo {year} {2024})},\ \bibinfo {note} {publisher: American Physical Society}\BibitemShut {NoStop}%
\bibitem [{\citenamefont {Zames}(1966)}]{zames_input-output_1966}%
  \BibitemOpen
  \bibfield  {author} {\bibinfo {author} {\bibfnamefont {G.}~\bibnamefont {Zames}},\ }\bibfield  {title} {\enquote {\bibinfo {title} {On the input-output stability of time-varying nonlinear feedback systems {Part} one: {Conditions} derived using concepts of loop gain, conicity, and positivity},}\ }\href {\doibase 10.1109/TAC.1966.1098316} {\bibfield  {journal} {\bibinfo  {journal} {IEEE Transactions on Automatic Control}\ }\textbf {\bibinfo {volume} {11}},\ \bibinfo {pages} {228--238} (\bibinfo {year} {1966})},\ \bibinfo {note} {conference Name: IEEE Transactions on Automatic Control}\BibitemShut {NoStop}%
\bibitem [{\citenamefont {Kuramoto}(1975)}]{kuramoto_self-entrainment_1975}%
  \BibitemOpen
  \bibfield  {author} {\bibinfo {author} {\bibfnamefont {Y.}~\bibnamefont {Kuramoto}},\ }\bibfield  {title} {\enquote {\bibinfo {title} {Self-entrainment of a population of coupled non-linear oscillators},}\ \ }(\bibinfo  {publisher} {Springer},\ \bibinfo {year} {1975})\ pp.\ \bibinfo {pages} {420--422}\BibitemShut {NoStop}%
\bibitem [{\citenamefont {Levine}(2018)}]{levine_control_2018}%
  \BibitemOpen
  \bibfield  {author} {\bibinfo {author} {\bibfnamefont {W.~S.}\ \bibnamefont {Levine}},\ }\href@noop {} {\emph {\bibinfo {title} {The control systems handbook: control system advanced methods}}}\ (\bibinfo  {publisher} {CRC press},\ \bibinfo {year} {2018})\BibitemShut {NoStop}%
\bibitem [{\citenamefont {Bechhoefer}(2021)}]{bechhoefer_control_2021}%
  \BibitemOpen
  \bibfield  {author} {\bibinfo {author} {\bibfnamefont {J.}~\bibnamefont {Bechhoefer}},\ }\href {\doibase 10.1017/9780511734809} {{\selectlanguage {en}\emph {\bibinfo {title} {Control {Theory} for {Physicists}}}}},\ \bibinfo {edition} {1st}\ ed.\ (\bibinfo  {publisher} {Cambridge University Press},\ \bibinfo {year} {2021})\BibitemShut {NoStop}%
\bibitem [{\citenamefont {Zhou}(1998)}]{zhou_essentials_1998}%
  \BibitemOpen
  \bibfield  {author} {\bibinfo {author} {\bibfnamefont {K.}~\bibnamefont {Zhou}},\ }\href {https://sharif.edu/~namvar/index_files/robust.pdf} {\emph {\bibinfo {title} {Essentials of {Robust} {Control}}}}\ (\bibinfo  {publisher} {Prentice Hall},\ \bibinfo {year} {1998})\BibitemShut {NoStop}%
\bibitem [{\citenamefont {Woolcock}\ and\ \citenamefont {Schmid}(2023)}]{woolcock_mixed_2023}%
  \BibitemOpen
  \bibfield  {author} {\bibinfo {author} {\bibfnamefont {L.}~\bibnamefont {Woolcock}}\ and\ \bibinfo {author} {\bibfnamefont {R.}~\bibnamefont {Schmid}},\ }\bibfield  {title} {\enquote {\bibinfo {title} {Mixed {Gain}/{Phase} {Robustness} {Criterion} for {Structured} {Perturbations} {With} an {Application} to {Power} {System} {Stability}},}\ }\href {\doibase 10.1109/LCSYS.2023.3290183} {\bibfield  {journal} {\bibinfo  {journal} {IEEE Control Systems Letters}\ }\textbf {\bibinfo {volume} {7}},\ \bibinfo {pages} {3193--3198} (\bibinfo {year} {2023})},\ \bibinfo {note} {conference Name: IEEE Control Systems Letters}\BibitemShut {NoStop}%
\bibitem [{\citenamefont {Horn}\ and\ \citenamefont {Johnson}(2012)}]{horn_matrix_2012}%
  \BibitemOpen
  \bibfield  {author} {\bibinfo {author} {\bibfnamefont {R.~A.}\ \bibnamefont {Horn}}\ and\ \bibinfo {author} {\bibfnamefont {C.~R.}\ \bibnamefont {Johnson}},\ }\href {\doibase 10.1017/CBO9781139020411} {\emph {\bibinfo {title} {Matrix {Analysis}:}}},\ \bibinfo {edition} {2nd}\ ed.\ (\bibinfo  {publisher} {Cambridge University Press},\ \bibinfo {year} {2012})\BibitemShut {NoStop}%
\bibitem [{Note1()}]{Note1}%
  \BibitemOpen
  \bibinfo {note} {To see this, consider the stacked system of complexified $\protect \bm {x}$ and the complex conjugate $\protect \bm {x}^*$, which is unitarily equivalent to a real system of stacked $\Re \protect \bm {x}$ and $\Im \protect \bm {x}$.}\BibitemShut {Stop}%
\bibitem [{\citenamefont {Thümler}\ \emph {et~al.}(2023)\citenamefont {Thümler}, \citenamefont {Srinivas}, \citenamefont {Schröder},\ and\ \citenamefont {Timme}}]{thumler_synchrony_2023}%
  \BibitemOpen
  \bibfield  {author} {\bibinfo {author} {\bibfnamefont {M.}~\bibnamefont {Thümler}}, \bibinfo {author} {\bibfnamefont {S.~G.}\ \bibnamefont {Srinivas}}, \bibinfo {author} {\bibfnamefont {M.}~\bibnamefont {Schröder}}, \ and\ \bibinfo {author} {\bibfnamefont {M.}~\bibnamefont {Timme}},\ }\bibfield  {title} {\enquote {\bibinfo {title} {Synchrony for {Weak} {Coupling} in the {Complexified} {Kuramoto} {Model}},}\ }\href {\doibase 10.1103/PhysRevLett.130.187201} {\bibfield  {journal} {\bibinfo  {journal} {Physical Review Letters}\ }\textbf {\bibinfo {volume} {130}},\ \bibinfo {pages} {187201} (\bibinfo {year} {2023})},\ \bibinfo {note} {publisher: American Physical Society}\BibitemShut {NoStop}%
\bibitem [{\citenamefont {Lee}\ \emph {et~al.}(2024)\citenamefont {Lee}, \citenamefont {Braun}, \citenamefont {Bönisch}, \citenamefont {Schröder}, \citenamefont {Thümler},\ and\ \citenamefont {Timme}}]{lee_complexified_2024}%
  \BibitemOpen
  \bibfield  {author} {\bibinfo {author} {\bibfnamefont {S.}~\bibnamefont {Lee}}, \bibinfo {author} {\bibfnamefont {L.}~\bibnamefont {Braun}}, \bibinfo {author} {\bibfnamefont {F.}~\bibnamefont {Bönisch}}, \bibinfo {author} {\bibfnamefont {M.}~\bibnamefont {Schröder}}, \bibinfo {author} {\bibfnamefont {M.}~\bibnamefont {Thümler}}, \ and\ \bibinfo {author} {\bibfnamefont {M.}~\bibnamefont {Timme}},\ }\bibfield  {title} {\enquote {\bibinfo {title} {Complexified synchrony},}\ }\href {\doibase 10.1063/5.0205897} {\bibfield  {journal} {\bibinfo  {journal} {Chaos: An Interdisciplinary Journal of Nonlinear Science}\ }\textbf {\bibinfo {volume} {34}},\ \bibinfo {pages} {053141} (\bibinfo {year} {2024})}\BibitemShut {NoStop}%
\bibitem [{\citenamefont {Bergen}\ and\ \citenamefont {Hill}(1981)}]{bergen_structure_1981}%
  \BibitemOpen
  \bibfield  {author} {\bibinfo {author} {\bibfnamefont {A.}~\bibnamefont {Bergen}}\ and\ \bibinfo {author} {\bibfnamefont {D.}~\bibnamefont {Hill}},\ }\bibfield  {title} {{\selectlanguage {en}\enquote {\bibinfo {title} {A {Structure} {Preserving} {Model} for {Power} {System} {Stability} {Analysis}},}\ }}\href {\doibase 10.1109/TPAS.1981.316883} {\bibfield  {journal} {\bibinfo  {journal} {IEEE Transactions on Power Apparatus and Systems}\ }\textbf {\bibinfo {volume} {PAS-100}},\ \bibinfo {pages} {25--35} (\bibinfo {year} {1981})}\BibitemShut {NoStop}%
\bibitem [{\citenamefont {Filatrella}, \citenamefont {Nielsen},\ and\ \citenamefont {Pedersen}(2008)}]{filatrella_analysis_2008}%
  \BibitemOpen
  \bibfield  {author} {\bibinfo {author} {\bibfnamefont {G.}~\bibnamefont {Filatrella}}, \bibinfo {author} {\bibfnamefont {A.~H.}\ \bibnamefont {Nielsen}}, \ and\ \bibinfo {author} {\bibfnamefont {N.~F.}\ \bibnamefont {Pedersen}},\ }\bibfield  {title} {{\selectlanguage {en}\enquote {\bibinfo {title} {Analysis of a power grid using a {Kuramoto}-like model},}\ }}\href {\doibase 10.1140/epjb/e2008-00098-8} {\bibfield  {journal} {\bibinfo  {journal} {The European Physical Journal B}\ }\textbf {\bibinfo {volume} {61}},\ \bibinfo {pages} {485--491} (\bibinfo {year} {2008})}\BibitemShut {NoStop}%
\end{thebibliography}%

\appendix

\section{Constant rank Lemma}
\label{lem:const rank}
\begin{lem}
    Let $\bm M_1(s),...,\bm M_N(s)$ be a set of matrix valued functions. 
    Assume that for any $s$, all the $\bm M_i(s)$ are semi-sectorial, have full rank, and their maximum and minimum phases satisfy 
    \begin{align}
        \max_i\overline{\phi}\left({\bm M}_l(s)\right) - \min_i\underline{\phi}\left({\bm M}_l(s)\right) &< \pi\, ,
    \end{align}
    Then
    \begin{align}\label{eq:kernels}
        \ker\left(\bm B \bigoplus_i \bm M_i(s) \bm B^\dagger\right) &= \ker\left(\bm B^\dagger\right)\, .
    \end{align}
    In particular, the matrix $    \bm B\bigoplus_i \bm M_i(s) \bm B^\dagger$ has constant rank as a function of $s$. 
\end{lem}

\begin{proof}
    For semi-sectorial matrices satisfying the minimum-maximum angle condition above, $0$ can only be on a corner of the numerical range, and thus has to be an eigenvalue. Thus, for such matrices, being full rank implies sectoriality and $0$ can not be in their numerical range. Further, all the numerical ranges of the $\bm M_i(s)$ at fixed $s$ lie in the same open half plane.

    Let us define $\bm M(s) \coloneqq \bigoplus_iM_i(s)$. 
    First, it is clear that $\ker \left( \bm B \bm M(s) \bm B^\dagger \right) \subset \ker \left( \bm B^\dagger \right)$.
    Indeed, if $\bm B^\dagger{\bm z} = 0$, then $\bm B \bm M(s) \bm B^\dagger{\bm z} = 0$. 

    Second, let ${\bm z}\notin\ker\left(\bm B^\dagger\right)$, i.e., $\bm B^\dagger{\bm z} \neq 0$. 
    Following assumption that all $\bm M_i(s)$ have their numerical ranges in the same open half-plane, $\bm M(s)$ has full rank and then $\bm M(s)B^\dagger{\bm z} \neq 0$. 
    Furthermore, by the convex hull property, the numerical range of $\bm M(s)$ lies in an open half-plane, and therefore
    \begin{align}
        {\bm z}^\dagger \bm B \bm M(s) \bm B^\dagger {\bm z} &\neq 0\, ,
    \end{align}
    meaning that $\bm B \bm M(s) \bm B^\dagger{\bm z}$ cannot be zero. 
    Hence, ${\bm z}$ does not belong to the kernel of $\bm B \bm M(s)\bm B^\dagger$, which concludes the proof of Eq.~\eqref{eq:kernels}. 

    Finally, as neither $\bm B$ nor its kernel depend on $s$, the  kernel is constant with respect to $s$

\end{proof}

\end{document}